\documentclass[sigconf]{acmart}

\definecolor{Gray}{gray}{0.85}
\newtheorem{lemma}{Lemma}
\def\name{DarKnight }
\usepackage{pifont}

\usepackage{mathtools}

\DeclarePairedDelimiter\floor{\lfloor}{\rfloor}
\usepackage{microtype}
\usepackage{hyperref}
\usepackage{url}
\usepackage{graphicx}
\usepackage{multirow}
\usepackage{algorithm}
\usepackage[noend]{algpseudocode}
\usepackage{subcaption}
\usepackage{colortbl}

\AtBeginDocument{%
  \providecommand\BibTeX{{%
    \normalfont B\kern-0.5em{\scshape i\kern-0.25em b}\kern-0.8em\TeX}}}

\copyrightyear{2021}
\acmYear{2021}
\setcopyright{rightsretained}
\acmConference[MICRO '21]{MICRO-54: 54th Annual IEEE/ACM International Symposium on Microarchitecture}{October 18--22, 2021}{Virtual Event, Greece}
\acmBooktitle{MICRO-54: 54th Annual IEEE/ACM International Symposium on Microarchitecture (MICRO '21), October 18--22, 2021, Virtual Event, Greece}\acmDOI{10.1145/3466752.3480112}
\acmISBN{978-1-4503-8557-2/21/10}

\begin{document}

\title{DarKnight: An Accelerated Framework 
           for Privacy and Integrity Preserving Deep Learning Using Trusted Hardware}


\author{Hanieh Hashemi}
\affiliation{%
  \institution{University of Southern California}
  \city{Los Angeles}
  \state{California}
  \country{US}}
\email{Hashemis@usc.edu}

\author{Yongqin Wang}
\affiliation{%
  \institution{University of Southern California}
  \city{Los Angeles}
  \state{California}
  \country{US}}
\email{Yongqin@usc.edu}

\author{Murali Annavaram}
\affiliation{%
  \institution{University of Southern California}
  \city{Los Angeles}
  \state{California}
  \country{US}}
\email{Annavara@usc.edu}

\renewcommand{\shortauthors}{Trovato and Tobin, et al.}

\begin{abstract}
Privacy and security-related concerns are growing as machine learning reaches diverse application domains. The data holders want to train or infer with private data while exploiting accelerators, such as GPUs, that are hosted in the cloud. Cloud systems are vulnerable to attackers that compromise the privacy of data and integrity of computations. Tackling such a challenge requires unifying theoretical privacy algorithms with hardware security capabilities. This paper presents DarKnight, a framework for large DNN training while protecting input privacy and computation integrity. DarKnight relies on cooperative execution between trusted execution environments (TEE) and accelerators, where the TEE provides privacy and integrity verification, while accelerators perform the bulk of the linear algebraic computation to optimize the performance. In particular, DarKnight uses a customized data encoding strategy based on matrix masking to create input obfuscation within a TEE. The obfuscated data is then offloaded to GPUs for fast linear algebraic computation. DarKnight's data obfuscation strategy provides provable data privacy and computation integrity in the cloud servers. While prior works tackle inference privacy and cannot be utilized for training, DarKnight's encoding scheme is designed to support both training and inference.  
\end{abstract}



\maketitle

\section{Introduction}
The need for protecting data privacy in deep learning is growing rapidly in many areas such as health care~\cite{blatt2020secure}, autonomous vehicles~\cite{zhu2014system}, finance~\cite{heaton2017deep}, communication technologies~\cite{foerster2016learning}. 

oudData holders in many of these application domains, however, are not machine learning experts. They rely on machine learning as a service (MLaaS) platforms such as Microsoft Azure ML~\cite{AzureML}, Google AI platform~\cite{Google}, Amazon ML~\cite{Amazon} to fulfill their ML needs. These services incorporate ML accelerators such as GPUs for high performance and provide easy to use ML runtimes to enable data holders to quickly set up their models and train. While these platforms lower the steep learning curve, they exacerbate the users' concern regarding data privacy. Typically it is not the case that the entire cloud service provider is untrustworthy, but a subset of cloud servers may be compromised by hackers and other adversaries as has been demonstrated in recent attacks~\cite{yu2019lagrange, dong2019secure}. Hence, our work is focused on protecting against a scenario where a subset of machines may be compromised, rather than a wholly untrusted cloud provider where every machine from the cloud provider  is compromised and colluding. Our goal is to protect
data privacy and computational integrity while still enabling the use of untrusted cloud systems. 
  
There are various algorithmic approaches to protect data privacy, such as Homomorphic Encryption libraries~\cite{gentry2009fully,liu2017oblivious,gilad2016cryptonets,juvekar2018gazelle}, Secure Multi-Party Computing (MPC)~\cite{shokri2015privacy, mohassel2017secureml, wagh2019securenn, mohassel2018aby3, ispass2022wang, tang2021verifiable}, Differential Privacy~\cite{abadi2016deep, erlingsson2014rappor,team2017learning}, Noise Injection~\cite{mireshghallah2020shredder,feng2022enhancing,mireshghallah2021not}, and using Trusted Execution Enviroments~\cite{narra2019privacy, tramer2018slalom}. Each of these methods provides a different privacy guarantee and comes at different cost~\cite{mirshghallah2020privacy}, as we explain in the next section. 

This work proposes DarKnight, a framework for accelerating privacy and integrity preserving deep learning while using untrusted accelerators. DarKnight uses a unique collaborative computing model between the hardware-supported trusted execution environments (\emph{TEE}) and GPU accelerators to tackle both privacy and performance challenges. The training/inference data from a client is only made visible within the TEE which prevents its visibility to an adversary. \name uses a novel data encoding strategy based on matrix masking to obfuscate input data within the TEE before the data is allowed to leave the TEE. The obfuscated data is then offloaded to GPUs to accelerate DNN's linear computations on the encoded data. Computing solely within TEEs can provide data privacy, by blocking access to TEE memory from intruders. Even though recent enhancements to TEEs support large memory workloads~\cite{SGXMemory}, TEE-enabled CPUs still do not support massive data parallel computations that accelerators such as GPUs and TPUs support. Therefore, DarKnight distributes compute-intensive linear operations to GPUs. DarKnight's usage of TEEs is limited to protecting the privacy of data through a customized matrix masking, verifying the integrity of computations, and performing non-linear operations (ReLU, Maxpool).

Table \ref{tab:GPUCMP} shows the relative speedup for various forward and backward propagation operations in the VGG16 DNN model~\cite{simonyan2014very} executing on a GPU (Nvidia GetForce GTX 1080 Ti GPUs) relative to a baseline model executing on Intel SGX (Coffee Lake E-2174G 3.80GHz CPU). Significant speedup can be achieved by offloading linear operations to GPUs. In the rest of the paper, we describe DarKnight's approach to collaboratively execute between TEE and GPUs to achieve high performance without compromising privacy. We also describe how \name can provide computational integrity by verifying computations performed on the GPUs.  
\begin{table}[htb]
\caption{Speedup in GPU relative to SGX in VGG16 Training on ImageNet. The baseline is implemented fully on Intel SGX}
\label{tab:GPUCMP}
\centering
\resizebox{\columnwidth}{!}{%
\begin{tabular}{c cccc}
\hline
Operations & Linear Ops  & Maxpool Time & Relu Time & Total \\ \hline
Forward Pass    & 126.85  & 11.86  & 119.60 & 119.03         \\
Backward Propagation  & 149.13  & 5.47 & 6.59 &  124.56 \\
\hline
\centering
\vspace{-6mm}
\end{tabular}
}
\end{table}

This paper makes the following major contributions:
\begin{itemize}
    \item We design a data privacy-preserving framework for DNN training and inference that uses a novel matrix masking scheme to encode data within the TEE.
    \item For high performance DNN training this work offloads linear computations on encoded data to be executed on GPUs. The novelty of the encoding scheme is such that linear computations on encoded data can be decoded within the TEE to accurately extract the necessary result.  
    \item We extend our encoding strategy to guarantee privacy even in the presence of colluding GPUs.
    
    \item We provide a low overhead mechanism for verifying computation integrity.
    
    \item We implemented DarKnight using an Intel SGX-enabled CPU to perform matrix masking and non-linear DNN operations while using an Nvidia GPUs to accelerate linear operations. We observe an average of $6.5$ x performance improvement for different DNN models. 
\end{itemize}

To the best of our knowledge,  this is the first work that uses TEE-GPU collaboration for \textbf{\emph{training}} large DNNs on private data while providing computational integrity and \textbf{\emph{provable}} data privacy.

The rest of the paper is organized as follows. In Section~\ref{sec:background}, we discuss related work and background. Section~\ref{sec:model} describes the system setting and DarKnight overview. We elaborate the encoding scheme in Section~\ref{sec:training}. In Section~\ref{sec:guarantee} privacy theorem is provided. Implementation and experimental results are presented in Section~\ref{sec:aggr} and Section~\ref{sec:setup}. In Section~\ref{sec:con}, we draw the conclusion. 
\section{Background and Related work}
 \newcolumntype{L}{>{\centering\arraybackslash}m{0.017\linewidth}} 
  \newcolumntype{D}{>{\arraybackslash}m{0.32\linewidth}} 
\begin{table*}[htb]
\caption{Comparison of applications and security guarantees of various prior techniques on neural networks' security ($\circ$ means method doesn't support that feature and $\bullet$ means it supports the feature)}
\label{tab:background}
\resizebox{\textwidth}{!}{%
\begin{tabular}{lcccccccccccc}
	\hline
	\hline
	\textbf{Method} & \textbf{Training} & \textbf{Inference} & \textbf{DP} & \textbf{MPC} & \textbf{HE} & \textbf{TEE} & \textbf{Data Privacy} &  \textbf{Model Privacy(Client)}&\textbf{Model Privacy(Server)}&\textbf{Integrity}&\textbf{GPU Acceleration}&\textbf{Large DNNs}\\
    \midrule
	SecureNN~\cite{wagh2019securenn} &$\bullet$ & $\bullet$ &  $\circ$& $\bullet$  & $\circ$&$\circ$&$\bullet$&$\bullet$&$\bullet$&$\circ$&$\bullet$&$\circ$\\
	Chiron~\cite{hunt2018chiron} &$\bullet$ & $\bullet$ & $\circ$	& $\circ$ & $\circ$&$\bullet$& $\bullet$&$\bullet$&$\bullet$& $\bullet$& $\circ$&$\circ$\\
	MSP~\cite{hynes2018efficient} &$\bullet$ &$\bullet$  & $\circ$	& $\circ$ &$\circ$ &$\bullet$& $\bullet$&$\bullet$&$\bullet$ &$\bullet$&$\circ$&$\circ$ \\

	Gazelle~\cite{juvekar2018gazelle} &$\circ$ &$\bullet$  &$\circ$  & $\circ$ & $\bullet$ &$\circ$&$\bullet$&$\circ$&$\circ$&$\circ$&$\bullet$&$\bullet$\\
	MiniONN~\cite{liu2017oblivious} &$\circ$ &$\bullet$  & $\circ$ & $\bullet$ & $\bullet$ &$\circ$&$\bullet$&$\bullet$&$\circ$&$\circ$&$\bullet$&$\bullet$\\
	CryptoNets~\cite{gilad2016cryptonets} &$\circ$ &$\bullet$  & $\circ$ & $\bullet$ & $\bullet$ &$\circ$&$\bullet$&$\bullet$&$\circ$&$\circ$&$\bullet$&$\bullet$\\
	Slalom~\cite{tramer2018slalom} &$\circ$ &$\bullet$  & $\circ$ & $\circ$ &$\circ$ & $\bullet$ & $\bullet$ &$\bullet$ &$\circ$&$\bullet$&$\bullet$&$\bullet$\\
	Origami~\cite{narra2019privacy} & $\circ$&$\bullet$  & $\circ$ & $\circ$ & $\circ$& $\bullet$ & $\bullet$ &$\circ$&$\circ$&$\circ$&$\bullet$&$\bullet$ \\
	Occlumency~\cite{lee2019occlumency} & $\circ$&$\bullet$  & $\circ$ & $\circ$ & $\circ$& $\bullet$ & $\bullet$&$\bullet$&$\bullet$&$\bullet$&$\circ$&$\bullet$ \\
	Delphi~\cite{mishra2020delphi} & $\circ$&$\bullet$  & $\circ$ & $\bullet$ & $\bullet$&$\circ$ & $\bullet$&$\bullet$&$\circ$&$\circ$&$\bullet$&$\bullet$ \\
	\textbf{DarKnight} & $\bullet$&$\bullet$  &$\circ$ & $\bullet$ &$\circ$ &$\bullet$ &$\bullet$  &$\bullet$ &$\circ$&$\bullet$ &$\bullet$&$\bullet$\\
	\hline
\end{tabular}
}
\end{table*}

\label{sec:background}
\subsection{Intel SGX}
\label{SGX}
TEEs such as ARMTrustZone~\cite{alves2004trustzone}, Intel SGX~\cite{costan2016intel}, and Sanctum~\cite{costan2016sanctum} provide a hardware-assisted secure execution environment where data privacy and computational integrity of the user's application is guaranteed by the hardware. Some of the cloud providers including IBM and Microsoft Azure have already equipped their cloud systems with Intel SGX~\cite{intel}. In this work, we also utilize Intel SGX for data privacy purposes. Intel SGX provides three important features: remote attestation, local attestation, and sealing that plays an important role in our scheme. TEEs generally provide a limited amount of secure memory that is tamper-proof. SGX provides $128$ MB as the enclave memory in its current implementation. Although recent enhancements to SGX have relaxed this memory limit, the performance of SGX still suffers when enclaves use very large memory due to the Merkle-tree based encryption and versioning requirements~\cite{SGXMemory, costan2016intel}.  
Furthermore, TEEs are CPU-based. Hence, there are not as well suited for massively parallel matrix operations that are common in DNNs. Although there are different methods proposed to extend the TEEs beyond CPUs, none of them are available in the marketplace yet~\cite{xu2020bus, volos2018graviton}. While some types of side-channel attacks have been performed on SGX~\cite{wang2017leaky,gotzfried2017cache,oleksenko2018varys, brasser2017software}, many of these attacks are being fixed actively~\cite{xu2015controlled}. 
As such SGX side channel leakage is outside the scope of this work. 

\subsection{Related Work}
There are a variety of approaches for protecting input and model privacy and computation integrity during DNN training and inference. These methods provide different privacy guarantees~\cite{mirshghallah2020privacy}. In this section, we briefly explain the most common methods. \textit{Fully Homomorphic encryption (FHE)} techniques encrypt input data and then perform operations directly on encrypted data. They usually provide a high theoretical privacy guarantee on data leakage, albeit with a significant performance penalty, and hence are rarely used in training DNNs. 
\textit{Secure multi-party computing (MPC)} is another approach, where multiple servers may use custom data exchange protocols to protect input data. 
They mostly use secret sharing schemes and have super-linear overhead for communication as the number of sharers and colluding entities grows. 
An entirely orthogonal approach is to use \textit{differential privacy (DP)}, which protects individual users' information through probabilistic guarantees by inserting noise signals to different parts of the computation. The tradeoff between utility and privacy is a challenge in this line of work. 
Protections to ML system sparse features are also discussed in \cite{laoram, pancakge2020}, though protections to sparse feature are not the focus of this paper.
Recently, TEEs attracted attention in ML domain for their privacy and integrity properties~\cite{asvadishirehjini2020goat,mo2020darknetz,park2020nested,hashemi2021byzantine, mo2021ppfl, ng2021goten,prakash2020mitigating}. Among TEE-based approaches,~\cite{tramer2018slalom} introduced Slalom an \emph{inference} framework that uses TEE-GPU collaboration to protect data privacy and integrity. However, as stated in their work their model was not designed for training DNNs. We elaborate on these reasons in Section~\ref{sec:litreture}. In another line of research, custom hardware accelerators are designed for security~\cite{volos2018graviton,hua2020guardnn}. These accelerator-based approaches use DNN-specific optimizations to reduce the protected memory footprint to improve efficiency. \textit{Instahide}~\cite{huang2020instahide} combines multiple images from a private dataset, and merges them with a public image set, and uses a sign flip function on pixels as random noise parameters. This method processes the encoded data without any decoding. However, privacy guarantees are purely empirical. That is why in~\cite{carlini2020attack} authors designed an attack to break this empirical guarantee. DarKnight provides a strong cryptography guarantee for privacy protection instead of relying just on empirical quantification which may be compromised. In Table~\ref{tab:background}, we compare some of these approaches based on their privacy and integrity guarantees, and their applications.

\section{DarKnight}
\label{sec:model}
\begin{figure}[htbp]
 \centering
 \includegraphics[scale = 0.2]{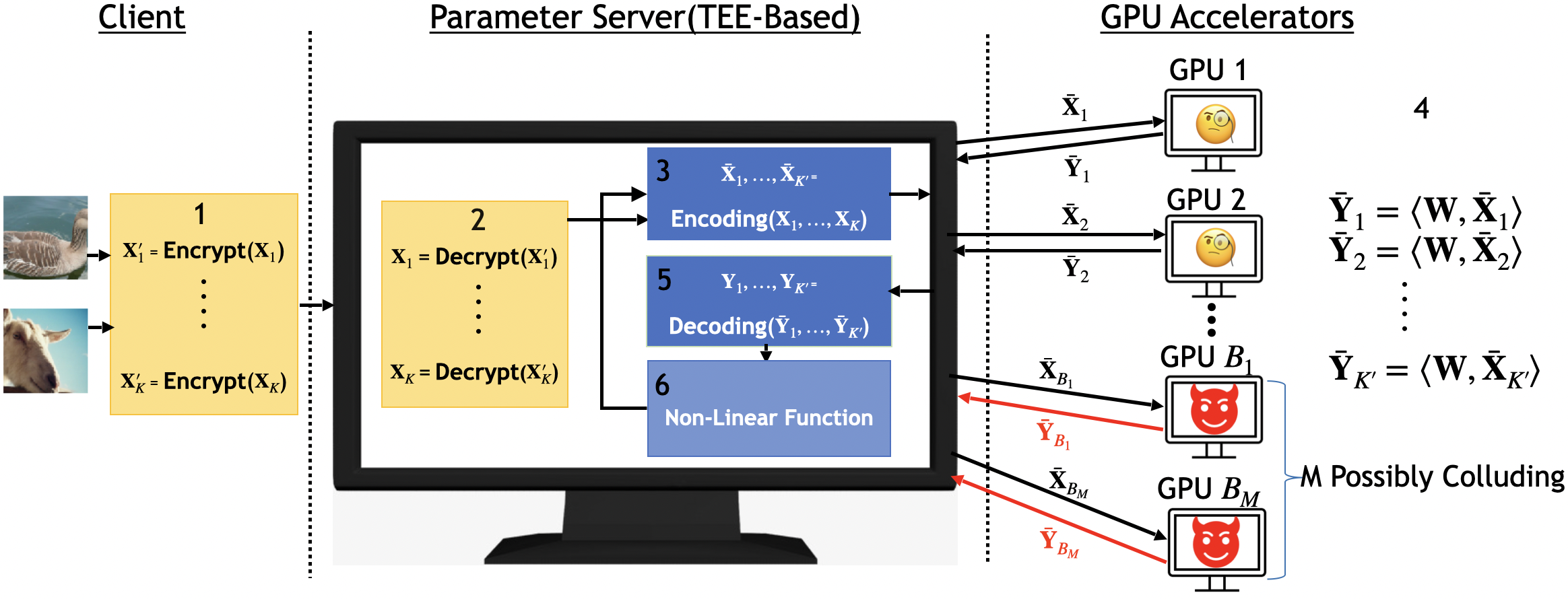}
5\vspace{-1mm}
 \caption{Overview of DarKnight components, its communication and, computation}
 \label{fig:model1}
\vspace{-5mm}
\end{figure}
\textbf{System Structure}: Our system model for learning is shown in Figure~\ref{fig:model1}. The client is the data owner who wishes to train a model using the cloud service. The cloud server is equipped with TEE which is responsible for protecting data privacy through encoding and performing non-linear operations. The TEE hardware guarantees code verification and authentication, where a data holder can verify the code and data usage within the enclave code. The system model uses $K'$ GPU accelerators ($\text{GPU}_1,\dots, \text{GPU}_{K'},{K'>1}$) that participate in linear computations on data that is encoded in the TEE. In our implementation, we use Intel SGX as a TEE. Communication channels between the client, server, and GPUs are encrypted. For example, all the client data is first encrypted before being sent to the TEE. A pairwise secure channel between TEE and each GPU can be established using a secret key exchange protocol at the beginning of the session. 
While GPUs perform most of the training computations, input privacy on GPUs is guaranteed by our proposed encoding scheme (described next), which obfuscates the original input. 

\textbf{Threat Model:} While adversaries can perform various attacks to exfiltrate DNN model parameters~\cite{riazi2019deep}, DarKnight focuses on attacks that expose the datasets used in training or inference and attacks that modify computational results on untrusted hardware.  We assume that data stored in TEE is protected from an adversary.  Side-channel attacks are out of the scope of this work.
The threat model on the cloud is a dynamic malicious adversary. This means at any given time,  accelerator GPUs may try to glean private information from the data shared with them by the TEE. To provide perfect privacy it is the responsibility of the TEE to provide the encoded data that is drawn from a uniformly random distribution which is decoupled from the actual raw input data (More details in Section~\ref{sec:guarantee}).
Moreover, a subset of GPUs may try to extract information by collaborating with each other. We refer to them as \emph{colluding GPUs}. Since the GPUs can be malicious, they may also inject faults in the computation to sabotage training or inference. Thus, computational integrity will be explicitly verified by \name in the presence of such an adversarial capability assumption.

Note that we assume only a subset of all the available GPUs collude. Such an assumption is not unique to this work as most MPC approaches assume only a subset of parties collude. Furthermore, this assumption is also true in scenarios where a client may request a subset of GPUs across different geographic regions, as is possible with current cloud services~\cite{rani2015ontology}. Clients may also request for nodes across multiple cloud providers to prevent collusion across the entire network. Nonetheless, it is useful to note that  DarKnight does not provide provable privacy guarantees against a scenario when all the GPUs participating in distributed training collude.   

In a system with $K'$ accelerator GPUs, DarKnight simultaneously provides:

\textbf{Data Privacy:} \name provides perfect privacy by making encoded data and raw inputs completely independent. Perfect privacy achieves when the data seen by any GPU does not give any information about the original input client data. More formally, perfect privacy is when the mutual information between encoded data and raw data is equal to zero. Namely, $I(\bar{X}_{k}:X_{k})=0$, where I is the mutual information~\cite{cover1999elements, yu2019lagrange}. 

\textbf{Integrity:} \name is (K'-1)-secure, namely it can \textit{detect} any malicious computation even if K'-1 GPUs send erroneous results to TEE.

\textbf{Collusion tolerance:} \name provide perfect privacy \textbf{and} integrity verification when $M$ GPUs collude, where $M$ is a function of K' (total GPUs in the system) and the number of inputs that can be encoded, as described at the end of this section~\ref{sec:training}.

\textbf{Model Privacy:}  \name does not reveal anything about the model to the client. However, model privacy on the server-side is out of the scope of this work. One common defense is using central differential privacy to keep the model private. Central differential privacy can be used on top of DarKnight~\cite{erlingsson2019amplification}.

\subsection{DarKnight Flow}
The initial machine learning model ($\mathbf{W}$) that a user wants to train is loaded into the cloud server and is made accessible to the untrusted GPUs as well. DarKnight then uses the following steps: 

\begin{enumerate}
    \item A batch of training/inference input data set is encrypted by the client and sent to the TEE enclave on the server. 
    
    \item TEE initiates encoding on decrypted inputs. 
    
    \item During the forward/backward pass of training, each layer requires linear and nonlinear operations. The linear operations are compute-intensive and will be offloaded to GPUs.  DarKnight's encoding mechanism is used to \textit{seal} the data before sending the data to GPU accelerators. To seal the data, DarKnight uses the notion of a \textit{virtual batch}, where $K$ inputs and a random noise are linearly combined to form $K+M$ coded inputs, as described in the next section; $M$ is the collusion tolerance as described above. The size of the virtual batch is limited by the size of the TEE memory that is necessary to encode $K$ images, typically 4-8 images at a time. This virtual batch size may be different from the traditional batch size used in machine learning. 

    \item The encoded data is offloaded to GPUs for linear operation. Each GPU receives only one encoded data inputs.  

    \item GPUs perform linear operations on different encoded data sets and return the results to TEE. 

    \item The TEE decodes the received computational outputs using DarKnight's decoding strategy and then performs any non-linear operations within the TEE.

    \item This process is repeated both for forward pass and backward propagation of each layer.  
\end{enumerate}

\section{DarKnight: Privacy in Training}
\label{sec:training}
In this section, we provide an overview of the forward and backward propagation of deep neural network training and the novel encoding and decoding process that we designed for privacy and integrity verification.    

We first start with the encoding/decoding in forward pass which is the first phase in training and then explain how backward propagation works. Please note that forward pass and inference are similar in terms of encoding and decoding functions and hence the forward propagation strategy can be applied directly to inference only systems.

In the rest of this section, we assume that we have a DNN with $L$ layers which is being trained with a virtual batch of $K$ inputs, the model parameters $\mathbf{W}_{l}$ at layer $l$ are updated using the well known SGD process.

Also for simplicity, we first show how this mechanism works for a system in which GPUs are not colluding, and next we expand the encoding to support a system with $M$ colluding GPUs in section~\ref{sec:colluding}.

\subsection{Forward Pass}
 At a layer $l$ of the forward pass, we need to  compute $\mathbf y_l=\langle \mathbf W_l~,~\mathbf x_l\rangle$, where $\mathbf W_l$ and $\mathbf x_l$ represent the model parameters and inputs in layer $l$, and $\langle \cdot,\cdot\rangle$ corresponds to the bilinear operation at that layer (e.g. matrix product, convolution, etc.).  After the linear operation finishes, an activation function ($g(\cdot)$) creates the next layer input $\mathbf x_{l+1}=\text{g}(\mathbf y_l)$.  Within this context, DarKnight first receives a set of $K$ inputs $\mathbf x_0^{(1)},\dots,\mathbf x_0^{(K)}$ for a batch training from a client. Our goal is to perform linear calculations of $\mathbf y_0^{(1)}= \langle\mathbf W_0 , \mathbf x_0^{(1)}\rangle,\dots,\mathbf y_0^{(K)}=\langle\mathbf W_0, \mathbf x_0^{(K)}\rangle$ on the GPUs without exposing the raw inputs to the GPUs. Note that the subscript $0$ in all these variables refers to the first layer. At this point, we drop the subscript for a more clear notation. Also, we use the notation $\mathbf {\color{red}x}$ for the inputs that need to be protected and $\mathbf{\color{blue}\bar{x}}$ for the encoded inputs. DarKnight must protect ${\mathbf{\color{red} x^{(i)}_{l}}}$ for each layer of the DNN when the layer's operations are outsourced to GPUs. 
 
 \textbf{Key Insight:} The main idea behind DarKnight's privacy protection scheme is the fact that the most computationally intensive operator (such as convolutions) is \emph{bilinear}. Thus, instead of asking a GPU to calculate $\langle \mathbf W,\mathbf {\color{red} x^{(i)}}\rangle$, which exposes the inputs, DarKnight uses matrix masking to linearly combine the inputs and add a random noise to them. Due to the bilinear property, any linear operation on $K$ masked inputs can be recovered if there are $K$ different linear computations performed.
 
\textbf{Matrix Masking}:
Introduced by~\cite{cox1980suppression, cox1994matrix}, matrix masking scheme can be used for a variety of reasons such as noise addition, sampling, etc.
Matrix masking uses the general form  $\mathbf B \mathbf X \mathbf A + \mathbf C$ for protecting Matrix $\mathbf X$. In the above formula $\mathbf B, \mathbf A$, and $\mathbf C$ are called record transformation masks, attribute transformation masks, and displacing masks, respectively. Any of these additive and multiplicative masks can be used for encoding data based on the data privacy goal. Prior approaches have used different combinations of masks to protect matrix data~\cite{kim1986method, spruill1983confidentiality, yu2019lagrange}.   
Darknight's encoding scheme is a form of matrix masking that is optimized for DNN's linear operations. 
\begin{figure}[tbp]
 \centering
 \includegraphics[width=\linewidth]{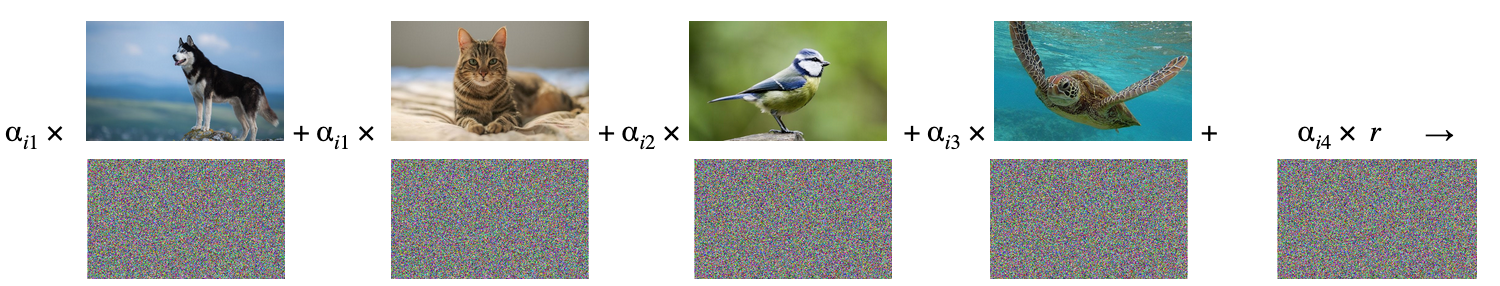}
 \caption{An example of DarKnight Encoding Scheme for K=4 }
 \vspace{-6mm}
 \label{fig:encoding}
\end{figure}

\textbf{DarKnight Encoding}: The SGX based enclave first receives a set of inputs (such as a set of images) from a data holder. Then the DarKnight scheme creates $K+1$ encoding within the SGX from $K$ data inputs (${\color{red}{\mathbf x}^{(1)}},\dots,{\color{red}{\mathbf x}^{(K)}}$), as follows,
\begin{align}\label{eq:inference_blinding}
{\color{blue}\bar{\mathbf x}^{(i)}}\quad= \quad \alpha_{1,i} {\color{red}\mathbf{x}^{(1)}} + \dots+ \alpha_{K,i} {\color{red}\mathbf{x}^{(K)}}  +\alpha_{(K+1),i} \mathbf{r}~
\end{align}
Where $i=1,\dots,(K+1)$. The scalars $\alpha_{i,j}$, and the noise vector $\mathbf r$ are randomly generated; and the size of $\mathbf r$ matches that of ${\color{red}\mathbf x}$.
The scalars $\alpha_{i,j}$'s are represented by matrix $\mathbf A \in \mathbb R^{(K+1),(K+1)}$, which are dynamically generated for each virtual batch and securely stored inside SGX for decoding. As we prove in a later section, by revealing the values ${\color{blue}\bar{\mathbf x}^{(i)}}$s to GPUs, we do not disclose any information about the inputs ${\color{red}\mathbf x^{(i)}}$s. An example of encoding is shown in Figure~\ref{fig:encoding}.

At the next step, each encoded data ${\color{blue}\bar{\mathbf x}^{(i)}}$ is sent to a $GPU_i$ which performs the following computation:
\begin{align*}
{\color{blue}\bar{\mathbf y}^{(i)}} =\langle \mathbf W , {\color{blue}\bar{\mathbf x}^{(i)}}\rangle, \quad i=1,\dots,(K+1).
\end{align*}
Please note that each GPU receives \emph{at most one} encoded data. 

\textbf{DarKnight Decoding}: The $K+1$ outputs ${\color{blue}\bar{\mathbf y}^{(i)}}$ returned from the GPUs must be decoded within the SGX to extract the original results ${\color{red}\mathbf y^{(i)}}$. These values can be extracted as follows:

\begin{align}
    {\color{blue}\bar{\mathbf Y}}=\left\langle \mathbf W, [{\color{blue}\bar{\mathbf x}^{(1)}},\dots,{\color{blue}\bar{\mathbf x}^{(K+1)}}] \right\rangle =
\end{align}
\begin{align*}
    \underbrace{\left\langle \mathbf W, [{\color{red}{\mathbf x}^{(1)}},\dots,{\color{red}\mathbf x^{(K)}},\mathbf r] \right\rangle}_{{\color{red}\mathbf Y}} ~\cdot \mathbf A~\Rightarrow~ {{\color{red}\mathbf Y}}={\color{blue}\bar{\mathbf Y}}\cdot \mathbf A^{-1}~
\end{align*}
\subsection{Backward Propagation}

The goal of training is to find the weight updates for each layer. For computing weight updates using SGD, we need to compute: 
\begin{equation}
\mathbf{W}^{\text{new}}_{l} = \mathbf{W}^{\text{old}}_{l} - \eta\times \triangledown \mathbf{W}_{l},\quad
\triangledown \mathbf{W}_{l}=\frac 1 K \sum_{i=1}^K ~ \langle \delta^{(i)}_{l} , {\mathbf x^{(i)}_{l}}\rangle
\label{eq:sgd}
\end{equation}
Here $x_l^{(i)}$ is the $i^{\text{th}}$ input of layer $l$. $\eta$ is the learning rate, and $\delta^{(i)}_{l}$ is the gradient of the loss for the $i^{\text{th}}$ point in the training batch, with respect to the output of layer $l$. 

There are two sets of computation intensive linear operations in each layer for computing backward propagation. \\
(1) The first linear operation is to compute the derivative of loss with respect to weights which is computed as follows: 
\begin{align*}
\triangledown \mathbf{W}_{l} =\frac{\partial\mathcal L}{\partial\mathbf W_l} = \frac {1}{K} \sum_{i=1}^K \langle \delta^{(i)}_l,  {\color{red}\mathbf{x}^{(i)}_l} \rangle 
\end{align*}

In this equation $\mathbf{x}^{(i)}_l$ needs to be protected for data privacy. Hence, when we want to offload these computations to GPUs, encoded data should be used instead of raw data. \\
(2) In the above equation $\delta^{(i)}_{l}$ represents the gradient of layer l for input data point $i$, which is computed as a second linear operation that computes the derivative of loss with respect to the output of the layer.   $\delta^{(i)}_{l}$ can be computed as follows, where $g_{l}$ is the activation function of layer $l$.

\begin{equation*}
\delta^{(i)}_{l} = \frac{\partial \ell(\mathbf{y}^{(i)}_{L}, \mathbf{y}^{(i)}_{exp})}{\partial \mathbf{y}_{l} } = \langle \delta^{(i)}_{l+1}~,~ {g^\prime_{l}}^{(i)} \rangle
\end{equation*}

As shown this computation does not contain any sensitive $\mathbf x^{(i)}$ information. Therefore, it can be offloaded to the GPUs without any data encoding.

In the rest of this section we explain how data encoding works for computing $\triangledown \mathbf W$. 

\textbf{Key Insight:} As explained in the previous section, the encoding process for forward pass exploited the invariant property of model parameter for any given input such  that $\left\langle \mathbf W, [{\color{blue}\bar{\mathbf x}^{(1)}},\dots,{\color{blue}\bar{\mathbf x}^{(k+1)}}] \right\rangle = \left\langle \mathbf W, [{\color{red}{\mathbf x}^{(1)}},\dots,{\color{red}\mathbf x^{(k)}},\mathbf r] \right\rangle ~\cdot \mathbf A~$, meaning that a single $\mathbf{W}$ was shared between all the inputs of that layers. However, during the backward propagation process, we a have different $\delta_l^{(i)}$ for each input $\mathbf {\color{red}x_l^{(i)}}$. Thus, \name designed a novel decoding scheme to extract the $\langle \delta^{(i)}_{l}, \mathbf {\color{red}x^{(i)}_{l}}\rangle$ from obfuscated inputs $\langle \delta^{(i)}_{l} , {\color{blue}\bar{\mathbf x}^{(i)}_{l}}\rangle$.

The decoding scheme is based on a key insight. While backward propagation operates on a batch of inputs, it is not necessary to compute the $\langle \delta^{(i)}_{l}, {\color{red}\mathbf x^{(i)}_{l}}\rangle$ for each input ${\color{red}\mathbf x^{(i)}}$. Instead, the training process only needs to compute \textit{cumulative parameter updates} for the entire batch of inputs. Hence, what is necessary to compute is the entire $\triangledown \mathbf{W}_{l}$ which is an average over all updates corresponding to inputs in the batch. 

\textbf{DarKnight Encoding:} DarKnight exploits this insight to protect privacy without significantly increasing the encoding and decoding complexity of the blinding process. 
As shown in Equation~\eqref{eq:sgd}, there are $K$ inputs on which gradients are computed. DarKnight calculates the overall weight update in the backward propagation by summing up the following $K+1$ equations each of which are computed on a different GPU,
\begin{align}\label{eq:gamma_lin}
\triangledown \mathbf{W} = \frac{1} {K} \sum_{j=1}^{K+1}  \gamma_{j} \text{Eq}_{j}, \qquad \text{Eq}_{j} = \left\langle \sum_{i=1}^K \beta_{j,i}~ \mathbf \delta^{(i)}~,{\color{blue}\bar{\mathbf x}^{(j)}} \right\rangle~
\end{align}
In the above equations, the encoded input ${\color{blue}\bar{\mathbf x}^{(j)}}$ to a layer is the same that was already calculated during the forward pass using Equation~\eqref{eq:inference_blinding}. As a result, the TEE can simply reuse the forward pass encoding without having to re-compute. 
As shown $\delta^{(i)}$s are multiplied with the $\beta_{j,i}$ in the GPUs after which the GPUs compute the bi-linear operation to compute $\text{Eq}_{j}$. 

In the above computation DarKnight selects $\alpha_{i,j}$'s, $\beta_{j,i}$'s, and $\gamma_i$'s such that they follow a very specific mathematical property as shown in Equation~\ref{eq:matrix_relation1}. If the values are selected to follow this property,  the overall parameter updates $\triangledown \mathbf{W}$ can be decoded very efficiently by scaling each $GPU_i$'s result with a corresponding $\gamma_i$'s and summing up all the $K+1$ outputs.

\begin{equation}
    \mathbf B^T\cdot \boldsymbol \Gamma\cdot \mathbf A^T = \begin{bmatrix}1 & 0 & \dots & 0 & 0
  \\0 & 1 & 0 & \dots & 0
  \\\vdots & \ddots& \ddots & \ddots & \vdots
\\ 0 & \dots & 0 & 1 & 0\end{bmatrix}_{K \times (K+1)}
\label{eq:matrix_relation1}
\end{equation}
Assuming batch size is equal to $K$, the $\beta_{j,i}$ parameters used for scaling $\delta$ values is gathered in the $K+1$ by $K$ matrix, $\mathbf B$. $\alpha_{i,j}$'s are gathered in the $K+1$ by $K+1$ matrix $\mathbf A$, and $\gamma_i$'s form the diagonal of a $K+1$ by $K+1$ matrix $\Gamma$, that gives us the proper parameters for efficient decoding. Note that the TEE keeps matrix $\Gamma$ and $\mathbf A$ as secret within the enclave memory, while providing $\mathbf B$ to the GPUs so each GPU can compute $\text{Eq}_{j}$.

\textbf{DarKnight Decoding:} Given the constraint imposed on $\alpha_{j,i}$'s, $\beta_{j,i}$'s and $\gamma_i$'s the decoding process is trivially simple to extract $\triangledown \mathbf{W}$. It is easy to see that if the scalars $\alpha_{i,j}$'s, $\beta_{i,j}$'s and $\gamma_i$'s satisfy the relation~\eqref{eq:matrix_relation1}, we will have
\begin{align}
    \frac 1 K\sum_{j=1}^{K+1}  \gamma_{j} ~ \text{Eq}_{j}=\frac 1 K \sum_{i=1}^{K} ~ \langle \delta^{(i)}_{l} , {\color{red}\mathbf x^{(i)}_{l}}\rangle=\triangledown \mathbf{W}_l
    \label{eq:agg}
\end{align} 
In other words, the decoding process only involves calculating a linear combination of the values in Equation~\eqref{eq:gamma_lin}. In particular, the TEE enclave receives $\text{Eq}_{j}$ from each $GPU_i$ and multiplies that value with the secretly held $\gamma_i$. It then computes the average over all the $\text{Eq}_{j}$ values received to compute the overall weight update for the batch.  

In Equation~\eqref{eq:gamma_lin}, the computation of $\beta_{j,i}~ \mathbf \delta^{(i)}$ has no privacy implication. We use this encoding only to simplify the secure aggregation mechanism to compute the weight update $\triangledown \mathbf{W}$. Thus we do not need to protect matrix $\mathbf B$ in the enclave. 

Note that even though the above equations are computed over $K$ inputs in a virtual batch, it is possible for the SGX enclave to securely store multiple $\triangledown \mathbf{W}_l$ associated with multiple virtual batches that comprise the training batch. Once all the inputs in the training batch are processed the SGX enclave can do a single aggregation to generate a batch-wide weight update. Details are provided in Section~\ref{sec:aggr}.

\textbf{DarKnight Training Complexity:}  The size of the $\alpha$,  $\delta$, and  $\gamma$ matrices is just proportional to the square of the virtual batch size that is being processed at one time. Therefore, generating them for every batch has a negligible performance overhead. Even with 8-64 batch size, (commonly used in VGG training~\cite{canziani2016analysis, han2015deep, narra2019slack}) these scaling values are substantially smaller than the model parameters $\mathbf{W}$. Hence, the order complexity of encoding/decoding operations is much less than the linear operations ($\left\langle \mathbf W, \mathbf x \right\rangle$) in a DNN with millions of parameters. In addition to that, the process of decoding $K$ inputs with one random noise requires $K+1$ computations. During decoding, we extract $\mathbf W \cdot \mathbf r$, but that value is just dropped. Thus, DarKnight trades only $\frac 1 K$ additional computations to provide provably perfect privacy.
\subsection{Decoding Correctness Proof}
\label{sec:dec}
In this section we prove that the above described decoding approach correctly produces the aggregate weight update within the TEE when using Equation~\ref{eq:agg}. Recall that \name generates $\bar{\mathbf{X}}=\mathbf{X} \mathbf{A}$ within the  TEE. The TEE also generates the random matrix ${\mathbf{B}}$, $\boldsymbol \Gamma$ to satisfy Equation~\ref{eq:matrix_relation1}. The $\bar{\mathbf{X}}$ along with the ${\mathbf{B}}$ are  sent to the GPUs. The GPUs generate $\bar{\mathbf{\delta}} = \mathbf{\delta} \mathbf{B}^T$ and then compute $\bar{\mathbf Y}= \bar{\mathbf \delta}^T \bar{\mathbf{X}}$. 

\begin{align}
    &{\mathbf X}=\left[{\mathbf x}^1,\dots,{\mathbf x}^K, \mathbf r\right]\in\mathbb R^{N\times K+1}~,\nonumber\\
    &\bar{\mathbf X}=\left[\bar{\mathbf x}^1,\dots,\bar{\mathbf x}^{K+1}\right]\in\mathbb R^{N\times K+1}~,\nonumber\\
    &{\mathbf \delta}=\left[{\mathbf \delta}^1,\dots,{\mathbf \delta}^K\right]\in\mathbb R^{N\times K}~,\nonumber\\
    &\bar {\mathbf \delta}=\left[\bar{\mathbf \delta}^1,\dots,\bar{\mathbf \delta}^K\right]\in\mathbb R^{N\times (K+1)}~,\nonumber\\
    & {\boldsymbol \Gamma}=\text{Diag}\left({\mathbf \gamma}^1,\dots,{\mathbf \gamma}^{K+1}\right)\in\mathbb R^{(K+1)\times (K+1)}~,\nonumber\\
    &\text{and ,}~~\mathbf A\in\mathbb R^{K+1 \times K+1}~,~ \mathbf B\in\mathbb R^{K+1 \times K}.
\end{align}

The goal is to compute $\sum_{i=1}^K \delta^{(i)} \mathbf{x^{(i)}}$ by using Equation~\eqref{eq:agg}. We make use of the following algebraic equations,
\begin{align}
&Tr[\mathbf X] = \sum_{i=1}^K {X_{ii}}\nonumber ~\quad \text{{definition of Matrix Trace}}\\
&Tr[\mathbf{X} \mathbf{Y}] = Tr[\mathbf{Y} \mathbf{X}]\nonumber\\
&(\mathbf{A} \mathbf{B})^T = \mathbf{B}^T \mathbf{A}^T\nonumber\\
&(\mathbf{A}^T)^T = \mathbf{A}
\end{align}
Now, as long as the equation~\ref{eq:matrix_relation1} holds
We will have
\begin{align}
&\sum_{j=1}^{K + 1} \gamma_j \text{Eq}_j = \mathrm{Tr}(\boldsymbol \Gamma \bar{\delta}^T \bar{\mathbf{X}}) = \mathrm{Tr}(\boldsymbol\Gamma \mathbf{B} \mathbf{\delta}^T \mathbf{X} \mathbf{A}) = \nonumber\\& \mathrm{Tr}(\mathbf{A} \boldsymbol{\Gamma}\mathbf{B} \mathbf{\delta}^T \mathbf{X}) = \mathrm{Tr}((\mathbf{B}^T \boldsymbol{\Gamma} \mathbf{A}^T)^T . (\mathbf{\delta}^T \mathbf{X}))=\sum_{i=1}^K \langle \delta_l^{(i)},\mathbf x_l^{(i)}\rangle~,
\end{align}
which concludes the proof.

\subsection{Computational Integrity}
\label{sec:integrity}
Apart from protecting data privacy, DarKnight's encoding method can be extended to detect computational integrity violations by untrusted GPUs. In this case, the linear computations performed by GPUs must also be verified.  
Recall that DarKnight  creates $K+1$ encoded inputs ${\color{blue}\bar{\mathbf x}^{(1)}},\dots,{\color{blue}\bar{\mathbf x}^{(K+1)}}$ for $K$ original inputs. To provide integrity, DarKnight creates one additional linear combination of inputs (say ${\color{blue}\bar{\mathbf x}^{(K+2)}}$), using the same approach as in Equation~\eqref{eq:inference_blinding}. 
This leads to having $K+2$ linear equations for recovering $K+1$ unknowns. This additional equation allows us to verify the accuracy of each result ${\color{red}{\mathbf y}^{(i)}}$ by computing it redundantly at least twice using at least two sets of equations. If the results computed are not consistent across the two, one of the GPU cores may not function properly or their data is modified by an attacker. Once an integrity violation is detected, TEE may perform additional corrective action, such as executing on another GPU worker or perform additional redundant computations. But these actions are outside the scope of our current work.

\subsection{Colluding GPUs}
\label{sec:colluding}
In this section, we investigate the scenario in which multiple GPUs can collaborate to extract information from the encoded data. With $K'$ GPUs and virtual batch size of $K$, we can tolerate $M < K'-K$ colluding GPUs without compromising privacy. We show how we can securely outsource calculating $\langle\mathbf W,\mathbf x^{(i)}\rangle$, $i=1,\dots, K$, to the GPUs. We first create $S=M+K$ encoded data vectors, $\bar{\mathbf x}^i$, $i=1,\dots, S$, using $M$ noise vectors $\mathbf r^1,\dots,\mathbf r^M$, as follows.
\begin{align}\label{eq:encode_X}
    &\bar{\mathbf X}=\mathbf X\mathbf A_1+\mathbf R\mathbf A_2~,\quad\text{where ,}\nonumber\\
    &\bar{\mathbf X}=\left[\bar{\mathbf x}^1,\dots,\bar{\mathbf x}^{S}\right]\in\mathbb R^{N\times S}~,\nonumber\\
    &{\mathbf X}=\left[{\mathbf x}^1,\dots,{\mathbf x}^K\right]\in\mathbb R^{N\times K}~,\nonumber\\
    &{\mathbf R}=\left[{\mathbf r}^1,\dots,{\mathbf r}^M\right]\in\mathbb R^{N\times M}~,\nonumber\\
    &\text{and ,}~~\mathbf A_1\in\mathbb R^{K\times S}~,~~\mathbf A_2\in\mathbb R^{M\times S}~.
\end{align}
Here, the matrices $\mathbf A_1$ and $\mathbf A_2$ are the encoding coefficients similar to the initial scheme we used for DarKnight. 
Same as before, we can calculate the weight updates using the following equations:
\begin{equation}\label{eq:gamma_lin_colluding1}
\triangledown \mathbf{W} = \sum_{j=1}^{S}  \gamma_{j} \text{Eq}_{j}, \quad \text{Eq}_{j} = \left\langle \sum_{i=1}^{K} \beta_{j,i}~ \mathbf \delta^{(i)}~,{\color{blue}\bar{\mathbf x}^{(j)}} \right\rangle
\end{equation}
We now define
\begin{align}
    \mathbf A=\begin{bmatrix}
    \mathbf A_1\\
    \mathbf A_2
    \end{bmatrix}~,~\mathbf B=\begin{bmatrix}
    \beta_{j,i}
    \end{bmatrix}~,\quad \Gamma=\text{Diag}(\gamma_1,\dots,\gamma_S)
\end{align}
Now, it is easy to show that we can recover $\triangledown \mathbf W$ using Equation~\eqref{eq:gamma_lin_colluding1} if:
\begin{equation}
    \mathbf B^\intercal\cdot \boldsymbol \Gamma\cdot \mathbf A = \begin{bmatrix}1 & 0 & \dots & 0 & 0 & \dots & 0
  \\0 & 1 & 0 & \dots & 0 & \dots & 0
  \\\vdots & \ddots& \ddots & \ddots & \vdots & \ddots  
\\ 0 & \dots & 0 & 1 & 0 &\dots & 0\end{bmatrix}_{K \times S}
\label{eq:matrix_relation2}
\end{equation}

As explained with this data encoding scheme, for a virtual batch of size $K$, and in a system with $M$ colluding GPUs, we need $K+M$ GPUs for data privacy reasons. Also, we need one extra GPU for integrity checks which performs redundant computation to verify the results. 

In summary:
\emph{In a system with $K'$ GPUs and virtual batch size $K$, DarKnight can provide data privacy and computational integrity by tolerating up to $M$ colluding malicious GPUs, where $K+M+1 \leq K'$.}

\section{Privacy Guarantee}
\label{sec:guarantee}
In this section, we explain how we provide data privacy to DarKnight. As we explained in the previous section, for each virtual batch, each GPU receives only one encoded data for each layer. To provide data privacy, encoded data should be completely independent of the raw data. If we make every encoded data, $\bar{\mathbf x}_i~\quad \forall i \in [1,...,S]$, appear uniformly random in the Finite Field $\mathbb F_p$ (p is a large prime) that is representing it, the adversary cannot infer any information about raw data. This is due to lemma \ref{lem:privacy_g}, since the values $\bar x_i$s are simply a linear combination of the variables $x_i$'s plus uniform random variables over the operating field. More concretely, $I(\bar{\mathbf X}_{M}:\mathbf X_{K})=0$ and knowing $\bar x_i$ leaks no information about $x_i$ . If every encoded data each GPU receives appears uniformly random, it will be impossible for a GPU to distinguish between different raw input data by just observing the encoded data~(it's the formal definition of computational indistinguishably in~\cite{goldreich2007foundations}). In the other words to achieve this goal, we make sure every encoded data exposed to GPUs looks uniformly random over Finite Field $\mathbb F_p$ for a large prime $p$. Since each GPU observes only one encoded data from a virtual batch, this privacy guarantee is one of the strongest privacy guarantees in the field of cryptography which is known as the one-time pad and used in many of the prior works~\cite{tramer2018slalom, yu2019lagrange,deng2004secure}. For this, first, DarKnight quantizes all the inputs and the model weights and maps them to the field $\mathbb F_p$. 
In this work, we choose $l = 8$ and $p=2^{25}-39$ for ResNet, VGG and MobileNet models which is the largest prime with $25$ bits. 

\begin{lemma}\label{lem:privacy_g}
Assume that $x\in\mathbb F_P$ is a scalar in the field $\mathbb F_p$, and $z$ is a random variable, uniformly chosen over the field $\mathbb F_p$, and let $y=x+z$. Then
\begin{align}
     I(y ; x) = 0 ~.
\end{align}
In the other words, knowing $y$ leaks no information about $x$.
\end{lemma}
\begin{proof}
Since the random variable $z$ is uniform over the field, the random variables $x$ and $y$ will be independent. Therefore, their mutual information will be zero.
\end{proof}
we can apply Lemma \ref{lem:privacy_g} to our work, by replacing the scalar $x$ with the linear combination of $x_i$'s, and replacing the random variable $z$ with the linear combination of uniform random variables $R_i$'s (note that linear combination of uniform random variables in the field would be a uniform random variable).

\textbf{Quantization}: There are various quantization schemes proposed in literature~\cite{gupta2015deep,so2020byzantine}, which strive to preserve accuracy while training with quantized data. 
Algorithm~\ref{alg:quant} presents our quantization code. At the first step of the quantization, we convert all the values from the floating-point domain to the fixed-point domain. To do so, after selecting the number of fractional bits, $l$, we multiply all the values by $2^l$ and then call the \texttt{Round} function to convert the values to integer representation as shown in lines 12 to 18. At the next step, we use function \texttt{Field} to transfer the computations to Finite Field. To do so, for negative numbers, we first add $p$ to them to transfer every value to $\mathbb F_p$. After that, we compute the remainder of the computation to $p$. Please note that both $\mathbf X^i$s and $\mathbf W^i$s are multiplied by $2^l$ before linear operations. Hence, bias should be multiplied by $2^{2l}$ so that after the linear operation we will be able to convert everything back to floating-point representation correctly.

After the linear computation on GPUs, TEE receives $\bar{\mathbf Y}$. TEE then subtracts $p$ from all the elements larger than $\frac p 2$ to restore negative numbers. At the next step, the decoding function is called to cancel the noise from linear operation results ($\mathbf Y_q$). Finally, \texttt{Round}(${\mathbf Y_q}\times 2^{-l}$) is called and at the end by multiplying $2^{-l}$ to the output, we have the result of the desired linear operation. Now, the original floating-point format is recovered for computing Non-Linear operations in TEE.

Please note that quantizing the weight and bias values does not have a privacy implication. These values may be stored and quantized in the GPUs so that the linear operation operates correctly in the Finite Field. For ResNet and MobileNet DNNs, this static quantization shows a good accuracy because these DNNs have normalization layers that keep the values within a range. However, for VGG models a slightly different quantization is used to dynamically normalize the values of inputs and weights if they pass the limits as suggested in~\cite{yang2019swalp, ng2021goten}. We normalize the values by dividing them to the maximum absolute entry of the vector. 

\textbf{Random Numbers}: Random scalars ($\mathbf A, \mathbf B, \boldsymbol \Gamma$) and random vectors ($\mathbf R$) are generated in the Finite Field $\mathbb F_p$ and encoding/decoding computations are defined over Field $\mathbb F_p$.

\begin{algorithm}
\caption{The pseudo-code for quantization}\label{alg:quant}
\begin{algorithmic}[1]
\Procedure{Linear Operations}{$\mathbf{W}, \mathbf{X}, \mathbf{b}$}
    \State $\mathbf W_q$ = Field(Round($\mathbf W \times 2^l$))
    \State $\mathbf b_q$ = Field(Round($\mathbf b \times 2^{2l}$))
    \State $\mathbf X_q$ = Field(Round($\mathbf X \times 2^l$))
    \State Generate $\mathbf R$ and all the random scalars $\mathbf A, \mathbf B, \boldsymbol \Gamma$ in the field $F_p$.
    \State $\bar{\mathbf X}$ = Encode($\mathbf{X}, \mathbf r, \mathbf A$) in the Field using equation \eqref{eq:encode_X}. 
    \State GPUs compute $\bar{\mathbf Y} = \langle \mathbf W_q, \bar{\mathbf X}, \mathbf b\rangle$
    \State ${\mathbf Y_q}$ = \text{Field} $(\bar{\mathbf{Y}} \mathbf A^{-1})$
    \State $\mathbf Y = \text{Round}(\mathbf Y_q \times 2^{-l})\times 2^{-l}$
    \State \textbf{return} $\mathbf Y$
\EndProcedure

\Procedure{Round}{$\mathbf X$}
    \For {$\forall X_i \in \mathbf X$}
        \If {$(X_i - \floor{X_i} < 0.5)$}
            \State $ {X_i}^r \gets \floor{X_i}$
        \Else {}
            \State $ {X_i}^r \gets \floor{X_i} + 1 $
        \EndIf
    \EndFor
    \State \textbf{return} $\mathbf X^r$
\EndProcedure

\Procedure{Field}{$\mathbf X$}
\For {$\forall X_i \in \mathbf X$}
        \If {$({X_i}^f < 0)$}
            \State $ {X_i}^f \gets {X_i}^f + p$
        \EndIf
        \State ${X_i}^f \gets {X_i}\quad \text{mod}\quad p$ 
    \EndFor
    \State \textbf{return} $\mathbf X^r$
\EndProcedure

\end{algorithmic}
\end{algorithm}

\textbf{Colluding GPUs}: As we explained in Section~\ref{sec:colluding}, we are using $K'$
 GPUs among which at most $M$ collude. Therefore, we need at least $M$ noise vectors according to DarKnight's scheme as shown in Equation~\ref{eq:encode_X}.
Now, the matrix $RA_2$ is $N$ by $S$ dimensional. Assume that the matrix $A_2$ is full rank and the noises in the matrix $R$ are independent and uniform over the field $F_p$. Now for any subset $\mathcal I\subseteq [K']$ with $|\mathcal I|\leq M$ of GPUs that collide, they will share the equations $XA_{1,\mathcal I} + RA_{2,\mathcal I}$ (where $A_{\mathcal I}$ is a sub-matrix of $A$ whose columns are chosen from the set $\mathcal I$). Since $A_{2}$ is full-rank, any subset of its columns are also full rank. Therefore, there is no linear combination that can vanish the noise in the $M$ shared equations, and any linear combination of uniform random variables over a field is also uniform. Hence, no matter how the $M$ equations are combined linearly, the result of the combination will seem uniform over the field $F_p$ to GPUs' perspective.

\section{DarKnight Implementation}
\label{sec:aggr}
In this section we  provide the details on how to efficiently implement DarKnight on off-the-shelf hardware consisting of Intel SGX CPU attached with Nvidia GPUs.

\textbf{Encoded Data Storage During Forward Pass}: The encoded intermediate feature map ($\bar {\mathbf x}^{(i)}$) that is received by GPUs during the forward propagation is also used during backward propagation for computing the weight update, as shown in Equation~\ref{eq:gamma_lin}. Rather than requesting TEE to resend the encoded inputs during back propagation, our current implementation of \name stores these encoded inputs within the GPU memory.

\textbf{Large Batch Aggregation}:
One of the benefits of the decoding process employed by DarKnight's back propagation is that an aggregate weight update across all the inputs in a virtual batch is computed by the TEE. Hence, individual weights associated with each input are not revealed (and not even computed by the TEE). The aggregate weight update over a virtual batch is essentially a customized version of secure aggregation for weigh updates~\cite{bonawitz2017practical}. As explained in section \ref{sec:training}, model parameters ($\mathbf W$) are stored outside the enclave and are visible to GPUs. Once the aggregate weight update $\triangledown \mathbf W$ is computed in the TEE, that data is sent to the GPUs so they can updates the model. Some prior works~\cite{zhu2019deep,geiping2020inverting} have shown that  
 $\triangledown \mathbf{W}$ may leak some information about the intermediate features which may eventually leads to input leakage as a side channel information.

While  side channels are outside the scope of this work, we believe that in practice the side channel information leakage can be curtailed. In particular, this prior work~\cite{zhu2019deep} also observed that one solution to reduce leakage drastically is to increase the batch size over which the aggregate weight updates are computed. Thus, using a large batch size, 
one can eliminate nearly all the side channel leakage. \name can be easily adapted to expose the $\triangledown \mathbf{W}$ of a large batch to the GPU memory. Earlier, we introduced the notion of a \textit{Virtual Batch}, which is essentially the largest number of images that we can be processed at the same time and fits inside SGX. The size of the virtual batch is purely a limitation of SGX memory. 

In our implementation we enable large batch weight aggregation as follows. The TEE still computes $\triangledown \mathbf{W}$ at the granularity of a virtual batch. However, instead of updating the weights immediately, our implementation aggregates the $\triangledown \mathbf{W}^v$ of multiple virtual batches $v>1$ within the TEE without disclosing the weight updates to the GPUs. The TEE internally accumulates $\triangledown \mathbf{W}^v$ associated with all the virtual batches from a large batch. 

One challenge we faced with this implementation is that to store multiple $\triangledown \mathbf{W}^v$ for all the virtual batches inside SGX exceeds the memory limitation. To resolve this problem, after each virtual batch computation, we encrypt the pages storing $\triangledown \mathbf{W}^v$ and then write it back to the untrusted memory. Note that this encryption and eviction is outside the critical path. Once all the virtual batches are processed the TEE then incrementally reloads each of the $\triangledown \mathbf{W}^v$, decrypts them and creates an aggregate update. 

We further optimized the reloading and aggregation of $\triangledown \mathbf{W}^v$ using a sharding technique. Note that the aggregation is highly parallel task. We break down the $\triangledown \mathbf{W}^v$ into multiple shards before storing them to the untrusted memory. For example, each shard may be a set of DNN layers. Then during the reloading process we perform shard-wise aggregation and then send those updates to the GPUs to incrementally update the model parameters. This pipelined approach to shard-wise aggregation essentially eliminates all the performance penalties associated with large batch aggregation.      

\begin{algorithm}
\caption{The pseudo-code for parameters update}\label{alg:back}
\begin{algorithmic}[1]
\Procedure{Backward}{$\mathbf{W}$}
    \State $V \gets VirtualBatch.size()$
    \State $N \gets LargeBatch.size()$
    \State $M \gets \frac {N}{V}$
    \State Initialize memory Pointer 
	\For {$v=1,2,\ldots,M$}\Comment{for each virtual batch}
	   \For {$l=1,2,\ldots,L$}\Comment{for each layer}
        \State Compute $\triangledown \mathbf{W}^{v}_{l}$ 
        \EndFor
        \State $\triangledown \mathbf{W'}^{v}$ = Encrypt($\triangledown \mathbf{W}^{v}$)
        \State Pointer.append (Evict($\triangledown \mathbf{W'}^{v}$))
    \EndFor
    \State $\triangledown \mathbf{W} \gets$ UpdateAggregation(Pointer)
    \State $\mathbf{W}^{\text{new}} = \mathbf{W}^{\text{new}} - \eta\times \triangledown \mathbf{W}$
\State \textbf{return} $\mathbf{W^{\text{new}}}$
\EndProcedure
\Procedure{UpdateAggregation}{Pointer}
    \State $V \gets VirtualBatch.size()$
    \State $N \gets LargeBatch.size()$
    \State $M \gets \frac {N}{V}$
	\For {$v=1,2,\ldots,M$}
        \State $\triangledown \mathbf{W}^{v} = $ Decrypt(Pointer $ + v \times W.size()$)
        \State $\triangledown \mathbf{W} = \triangledown \mathbf{W} +\triangledown \mathbf{W}^{v}  $
    \EndFor
    \State \textbf{return} $\triangledown \mathbf{W}$
\EndProcedure

\end{algorithmic}
\end{algorithm}

\textbf{Pseudo Code of Aggregation:}
Algorithm~\ref{alg:back} shows the steps of backward propagation. In line 4, the actual batch size is divided to the virtual batch size to get the number of virtual batches we have to process before updating the weight. Line 5 initializes the pointer to the untrusted memory where we store $\triangledown \mathbf W$s before the aggregation. 
Line 6 repeats the for loop for each virtual batch.
Line 7 to 8 shows how $\triangledown \mathbf W$ is computed for each layer of DNN.
Line 9 Encrypt $\triangledown \mathbf{W}^{v}$ that is containing all the weight updates of the network for that virtual batch.  
Line 10 calls evict function containing $\triangledown \mathbf{W'}^{v}$. As a result, this page is moved to the untrusted memory with all the precautions. 
Line 11 calls Aggregation function. This function has a reference to all the pages containing part of the $\triangledown \mathbf{W}$ and construct the whole $\triangledown \mathbf{W}$.
Finally in Line 12 the $\triangledown \mathbf{W}$ is sent to GPUs for the weight updates.

In Figure~\ref{fig:impl}, we show how different sizes of virtual batch size can speedup the aggregation time. Having a larger virtual batch size helps with the less number of encryption/eviction and decryption. However, increasing a size of virtual batch at a certain point, will increase the latency because the whole data cannot fit inside the enclave memory. As depicted in this figure, if we want to update the weight for batches of size 128, for all our three networks, virtual batch size of 4 ($K=4$) shows the best performance. Note that this value of $K$ will likely increase in future SGX implementation as the protected memory segment size grows. But in our implementation on current Intel SGX systems K=4 shows the best performance.  

\begin{figure}[tbp]
 \centering
 \includegraphics[width=\linewidth]{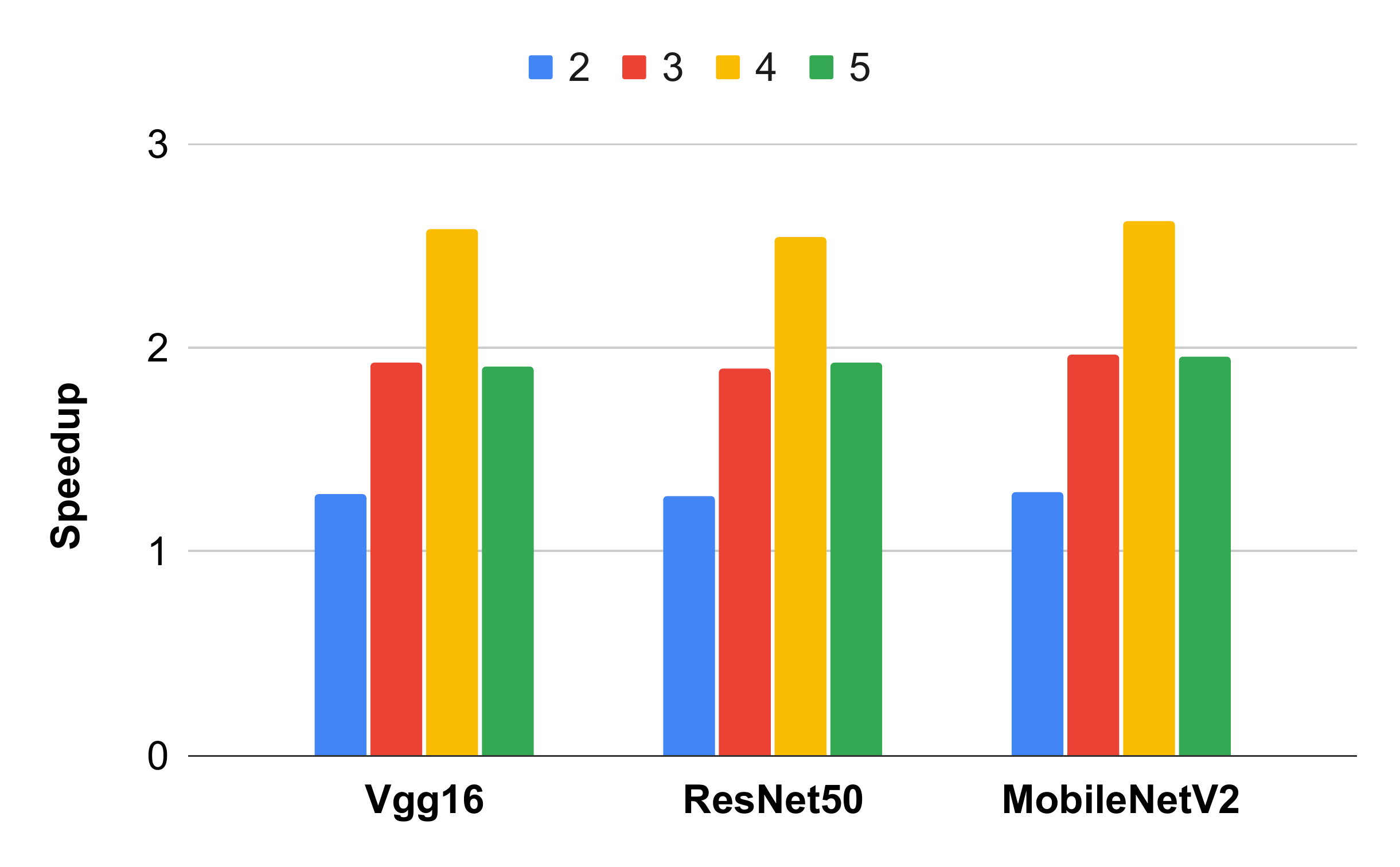}
 \caption{Effect of different virtual batch sizes ({number of GPUs(K')-1 when M = 0}) on the aggregation ({decoding}) speedup relative to virtual batch size K = 1}
  \vspace{-5mm}
 \label{fig:impl}
\end{figure}

\begin{figure*}[tbp]
  \centering
  \begin{subfigure}[b]{0.30\linewidth}
    \includegraphics[width=\linewidth]{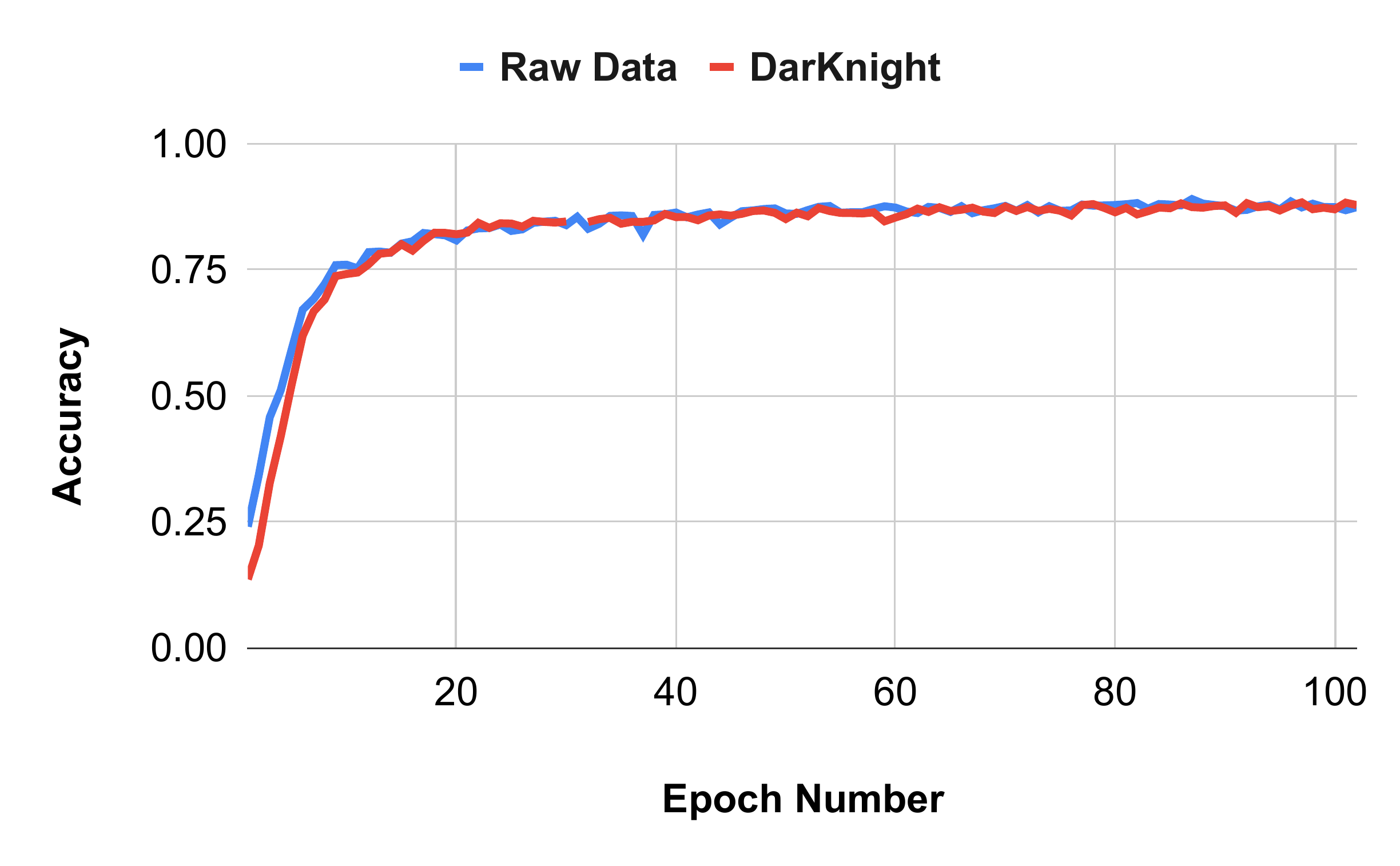}
  \end{subfigure}
  \begin{subfigure}[b]{0.30\linewidth}
    \includegraphics[width=\linewidth]{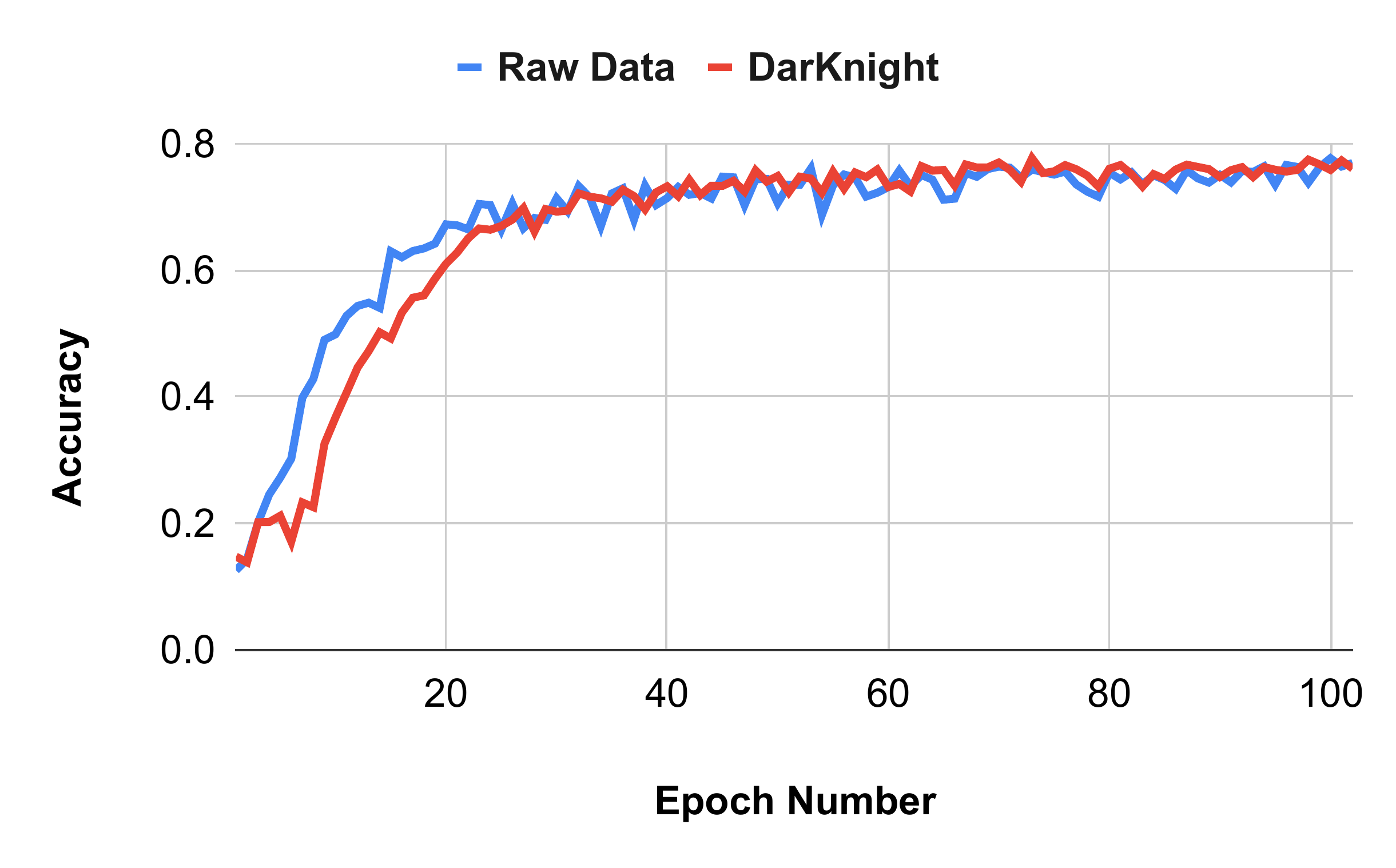}
  \end{subfigure}
   \begin{subfigure}[b]{0.30\linewidth}
    \includegraphics[width=\linewidth]{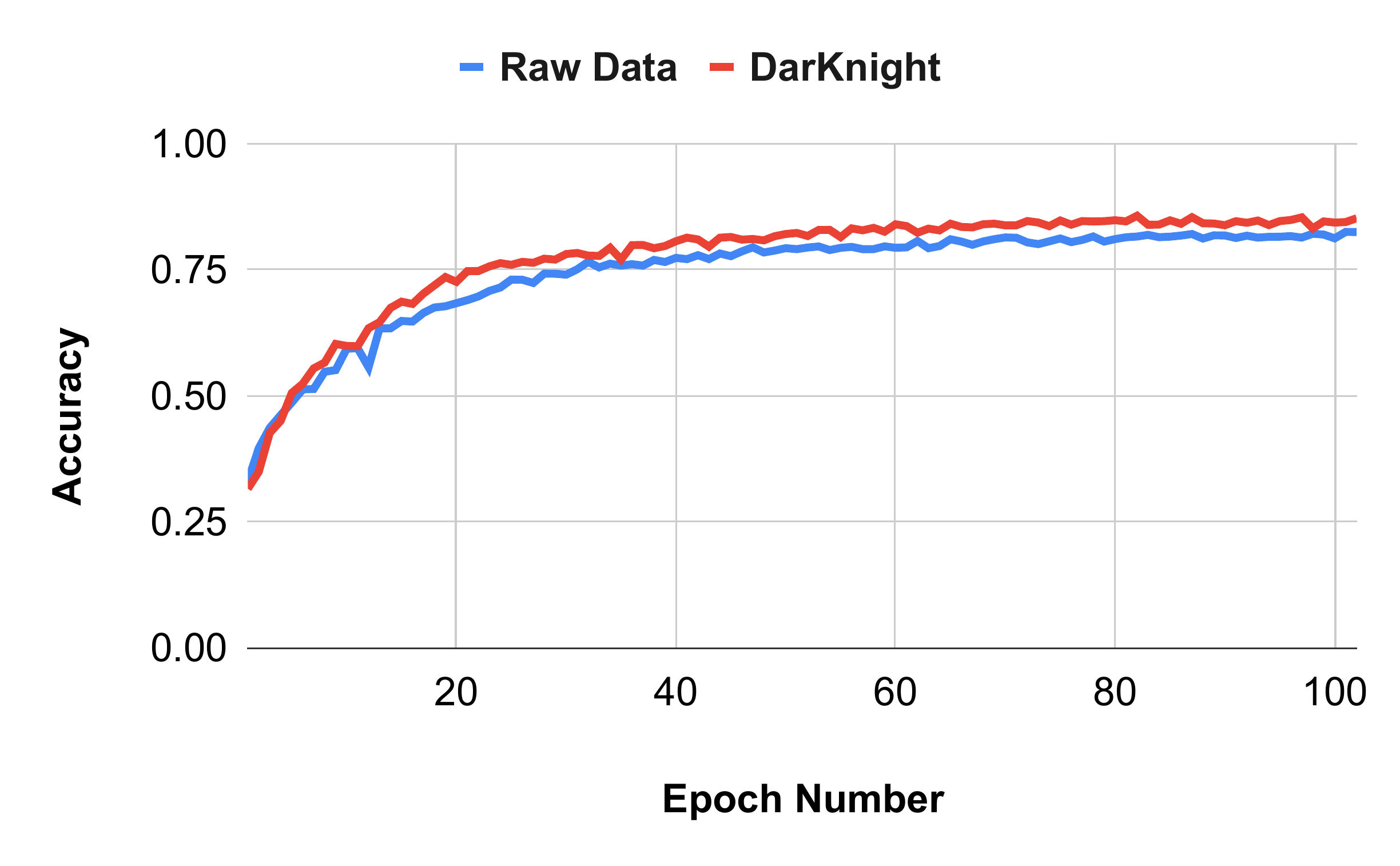}
  \end{subfigure}
 \vskip -0.15in
  \caption{Training accuracy of DarKnight for CIFAR-10 with (a) VGG16 (b) ResNet50 (c) MobileNetV2}
  \label{fig:training0}

\end{figure*}

\section{Experimental Setup and Results}
\label{sec:setup}
We implemented \name on a server consisting of an Intel Coffee Lake E-2174G 3.80GHz processor with SGX support, and Nvidia GeForce GTX 1080 Ti GPUs. The server has 64 GB RAM and supports Intel Soft Guard Extensions (SGX). DarKnight's training scheme and the related unique coding requirements are implemented as an SGX enclave thread where both the decoding and encoding are performed. For SGX implementations, we used Intel Deep Neural Network Library (DNNL) for designing the DNN layers including the Conv layer, ReLU, MaxPooling, and Eigen library for Dense layer. We used Keras 2.1.5, Tenseflow 1.8.0, and Python 3.6.8.

To evaluate the method, we used three different DNN models with different structures: VGG16~\cite{simonyan2014very} with 138 million parameters, ResNet50~\cite{he2016deep} with 23 million parameters and, MobileNetV2~\cite{sandler2018mobilenetv2} with 3.4 million parameters. We chose MobileNetV2 because it is the worst-case benchmark for our model as it reduces linear operations considerably (using depth-wise separable convolution), thereby reducing the need for GPU acceleration. 
We used CIFAR-10 \cite{krizhevsky2009learning} that has 50,000 training images evenly distributed between 10 categories, 
and ImageNet~\cite{russakovsky2015imagenet} with about 1.2 million images and 1000 categories as our datasets.

\subsection{Training Results}
For evaluating training performance, two critical aspects are examined: accuracy impact and speedup of the training.

\textbf{Effect of Random Noise and Quantization on Accuracy}: Quantization, using fixed-point arithmetic and adding the uniform random noise to encode the data for privacy reasons may cause accuracy degradation on GPUs. To study the impact, Figure~\ref{fig:training0} shows the training accuracy on clean data when there is no privacy scheme in the system and the accuracy of training in the presence of data encoding on VGG16, ResNet50, and MobileNetV2 on CIFAR-10. 
For training on CIFAR-10, as depicted in Figure~\ref{fig:training0} after 100 epochs, only a negligible accuracy degradation (less than 0.01 for all the applications) is observed. Very similar behavior is observed across a wide range of models.  
{For MobileNetV2, even a slight increase is observed. This behaviour is observed in prior works as quantization may remove the noise on data~\cite{lin2019defensive,lin2016fixed}.}

\begin{figure}[tbp]
    \includegraphics[scale=1,width=\linewidth]{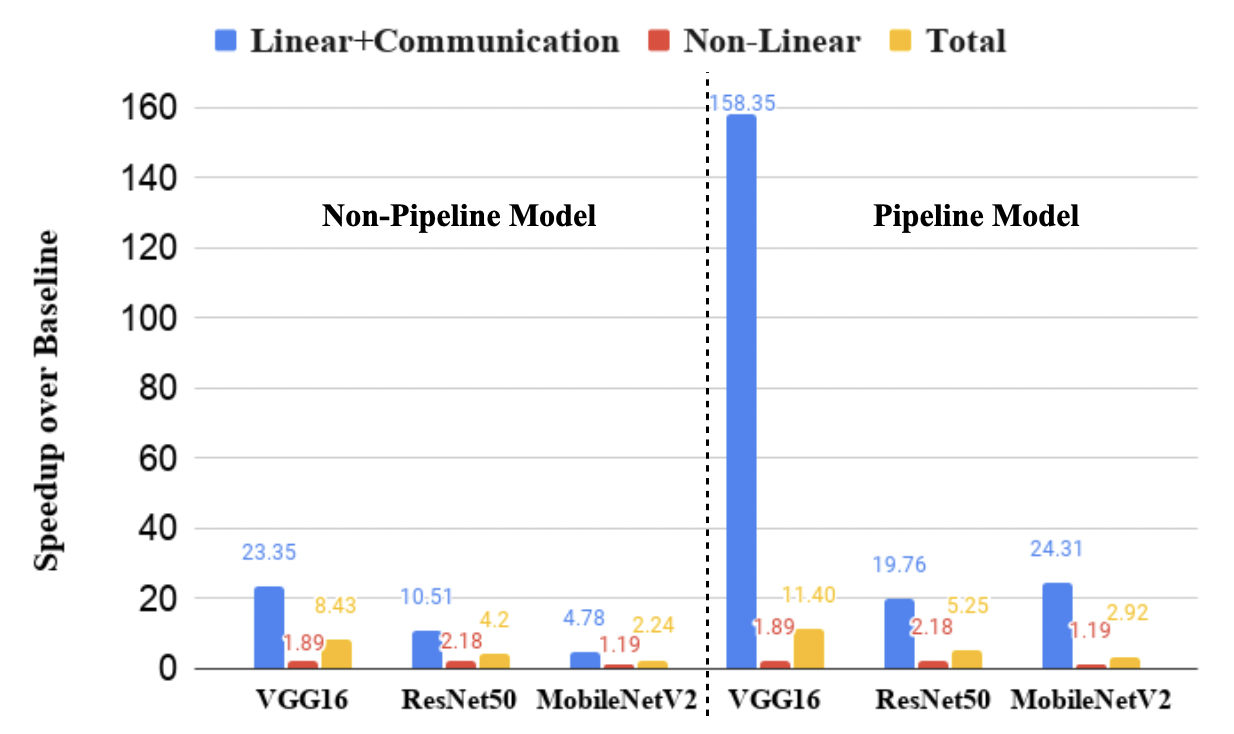}
    \caption{ImageNet Training Speedup for None-Pipeline and Pipeline Settings}
    \label{fig:trainingtime}
    \vspace{-4mm}
\end{figure}

\textbf{Training Execution Time, Non-Pipelined}: Figure~\ref{fig:trainingtime} shows the speedup of training using DarKnight with $K=2$ images encoded and offloaded to $3$ GPUs relative to the baseline fully implemented on SGX. The encoding operation in TEE and the linear operations in GPUs are serialized in this design, which we refer to as the non-pipelined version. The figure breaks down the execution time into linear and non-linear categories. Non-linear category includes all the operations performed in TEE including the encoding and decoding overheads, while linear costs include operations performed on GPUs plus the communication cost to move the encoded data to the GPUs and bring the computed results back to the TEE. For emulating the communication time we used Infiniband~\cite{shanley2003infiniband,pfister2001introduction} network switch between TEE and each GPU which has a 40 Gbps bandwidth. We assume a non-pipelined implementation where the data encoding and decoding process are performed sequentially. 

The results show that DarKnight speeds up the total linear operation time of VGG16 by $23$x  by exploiting GPUs parallelism. 
Even the non-linear operations see a $1.89X$ speedup. This result is due to the observation that the baseline has to encryption/decrypt some of the large intermediate feature maps that do not fit within SGX memory. 
Overall for VGG16 the execution time is improved by more than $8X$ with DarKnight. Both ResNet and MobileNet models have batch normalization layers which are non-linear operations that are computation-intensive and cannot be offload to GPU accelerators. Even in this worst-case scenario, performance gains of $4.2X$ and $2.2X$ are achieved.

Table~\ref{tab:training1} shows a more detailed breakdown of the fraction of time spent in various operations in each setting. In the baseline majority of time is spent in the linear operations. For VGG16 baseline spends $84\%$ of the time on linear operations. Due to batch-normalization overheads and reduced linear operation counts ResNet50 and MobileNetV2 spend around $60\%$ in linear execution time. Using the DarKnight the distribution of execution time is reversed. VGG16  spends nearly $50\%$ of the time on non-linear operations as GPUs accelerate linear operations. It also pays $19\%$ overhead for encoding and decoding. For ResNet50 and MobliNetV2 this overhead is lower because batch normalization dominates the execution time. DarKnight also pays a new communication overhead to move the data between TEE and GPUs. Across all three models, about 20\% of the total training time in DarKnight is spent in this communication phase.  
\begin{table}[tbp]
    \caption{ImageNet Training Time Breakdown for Different Networks ({values are percentage relative to the total execution time})}
    \label{tab:training1}
    \resizebox{\linewidth}{!}{%
\begin{tabular}{|l|cc|cc|cc|}
\hline
\multirow{2}{*}{Operation} & \multicolumn{2}{c|}{\textbf{VGG16}} & \multicolumn{2}{c|}{\textbf{ResNet50}} & \multicolumn{2}{c|}{\textbf{MobileNetV2}} \\ \cline{2-7} 
                           & DarKnight     & Basline    & DarKnight      & Baseline      & DarKnight        & Basline       \\ \cline{2-7} 
Linear                     & 0.04          & \cellcolor[gray]{0.6}0.84       & 0.04           & \cellcolor[gray]{0.6}0.61          & 0.06             & \cellcolor[gray]{0.6}0.62          \\
NonLinear                  & \cellcolor[gray]{0.8}0.50          & 0.16       & \cellcolor[gray]{0.8}0.75           & 0.39          & \cellcolor[gray]{0.8}0.63             & 0.38          \\
Encoding-Decoding                  & 0.19          & 0       & 0.01           & 0          & 0.08             & 0          \\
Communication              & 0.26          & 0          & 0.2            & 0             & 0.23             & 0             \\ \hline
\end{tabular}
}
 \vskip -0.2in
\end{table}

\textbf{Training Execution Time, Pipelined}: The non-pipelined implementation in the previous section assumes that the data encoding/decoding in TEE is serialized with the GPU executions. However, the communication overhead can be easily hidden by overlapping communication and computations. Meaning that, while data communication happens between GPUs and TEE, TEE can encode the next data batch of data. Pipelined implementation with asynchronous SGD has been designed in prior work~\cite{zhang2015staleness,narayanan2019pipedream}. A pipelined implementation encodes one virtual batch and launches it to the GPUs. While GPUs are performing linear operations, the next virtual batch is encoded under the shadow of GPUs execution time.  Figure~\ref{fig:trainingtime} shows the results. By Overlapping the communication and computation, DarKnight speeds up the total linear operation time by $20-158$x with pipelined design which leads to an overall higher speedup.

\begin{figure*}[htbp]
  \begin{subfigure}[b]{0.45\linewidth}
    \includegraphics[width=\linewidth]{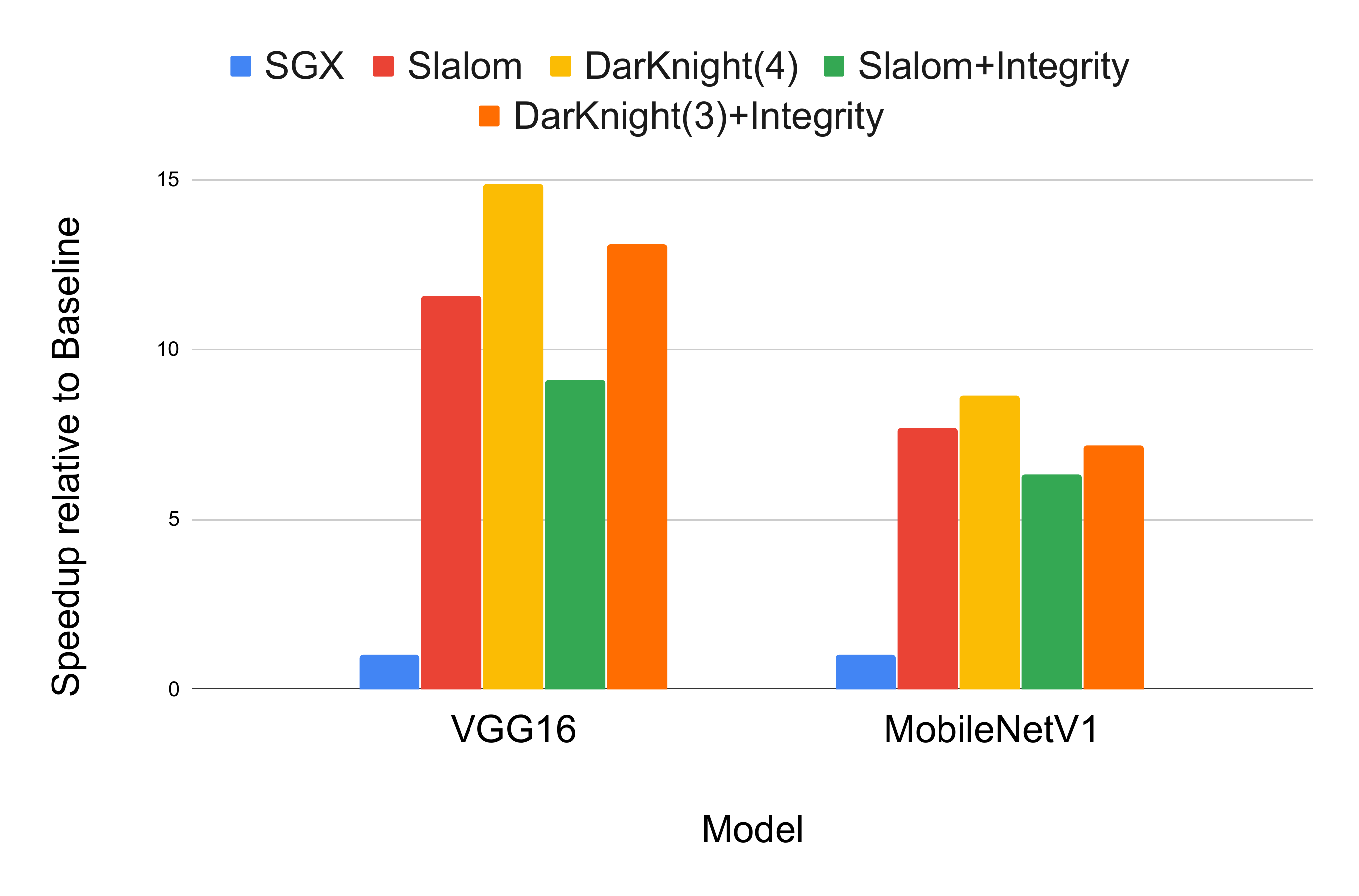}
     \label{fig:inference1}
  \end{subfigure}
  \begin{subfigure}[b]{0.45\linewidth}
      \includegraphics[width=\linewidth]{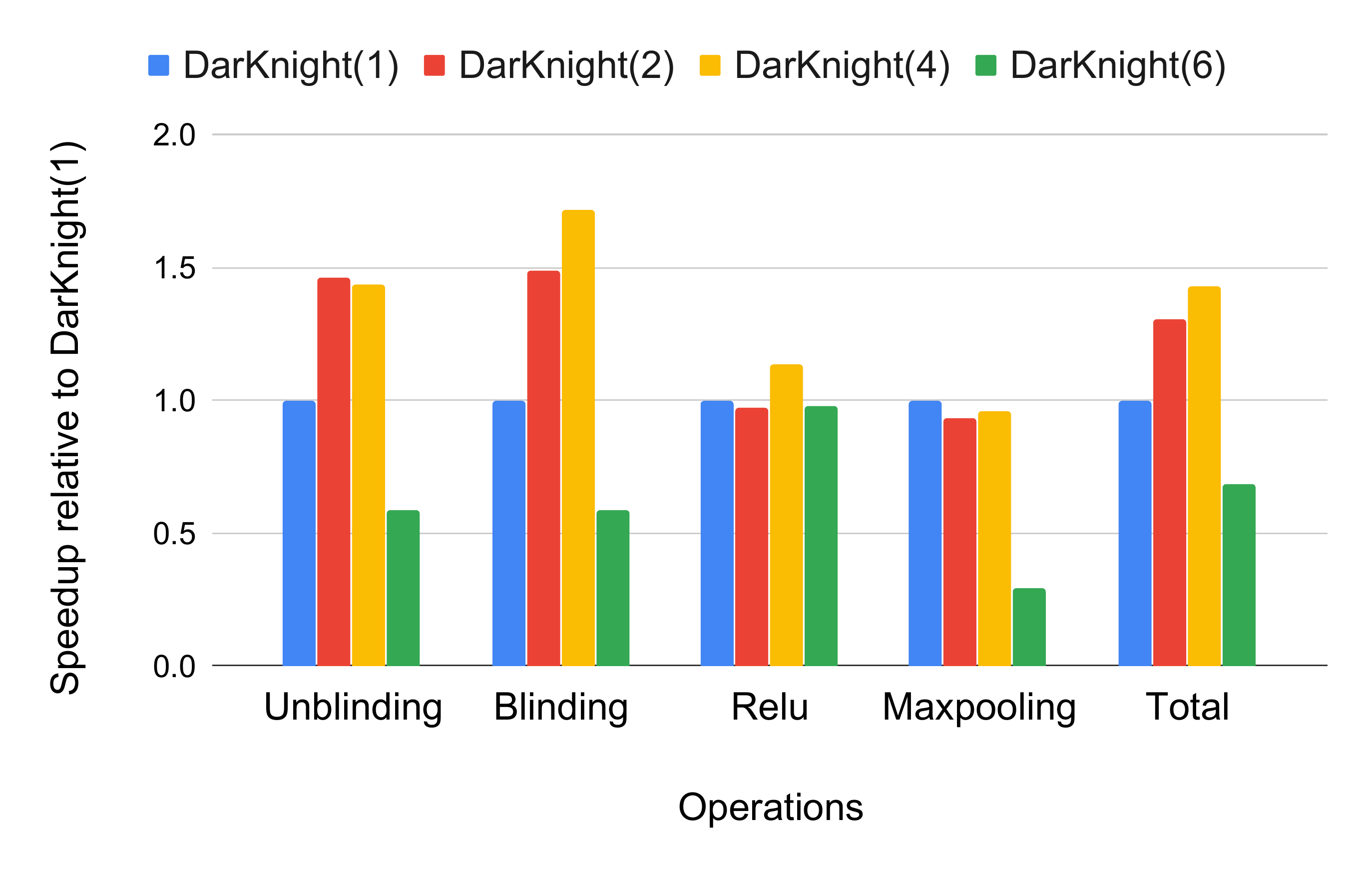}
   \label{fig:combine}
  \end{subfigure}
  \vspace{-9mm}
  \caption {a) Inference speedup comparison with different implementations relative to SGX for VGG16, and MobileNetV1 on ImageNet.
  b) Inference speedup comparison of different operations relative to DarKnight(1) for different virtual batch-sizes for VGG16 on ImageNet.
  }
   \vspace{-3mm}
  \label{fig:2}
\end{figure*}

{\textbf{Non-Private Model:} Table~\ref{tab:GPU3} shows the training speedup of using three unprotected GPUs with Pytorch library over SGX-only computation or using DarKnight (3 GPUs). However, unprotected GPUs does not provide any privacy guarantee. Using only SGX to provide privacy leads to two orders of magnitude slowdown. DarKnight is designed to bridge this gap between TEE and GPU speed while protecting the privacy of data. Please note that, other privacy preserving methods, such as fully homomorphic encryption and secure multi-party computing, are orders of magnitude slower than GPU implementation~\cite{juvekar2018gazelle,ghodsi2017safetynets}.}
\begin{table}[htb]
\vspace{-2mm}
\caption{{Non-Private Training Speedup on 3-GPUs Relative to SGX-only and DarKnight (3 GPUs) for Training on ImageNet} }
\label{tab:GPU3}
\centering
\resizebox{\columnwidth}{!}{%
\begin{tabular}{c ccc}
\hline
{Model} & {VGG16}   & {ResNet50} & {MobileNetV2} \\ \hline
{Speedup over DarKnight}    & $23.93$  & $41.01$  & $27.51$        \\
{Speedup over SGX}  & $273.26$  & $216.62$ & $80.31$  \\
\hline
\centering
\vspace{-8mm}
\end{tabular}
}
\end{table}

\begin{figure}[tbp]
 \centering
 \includegraphics[width=0.7\linewidth]{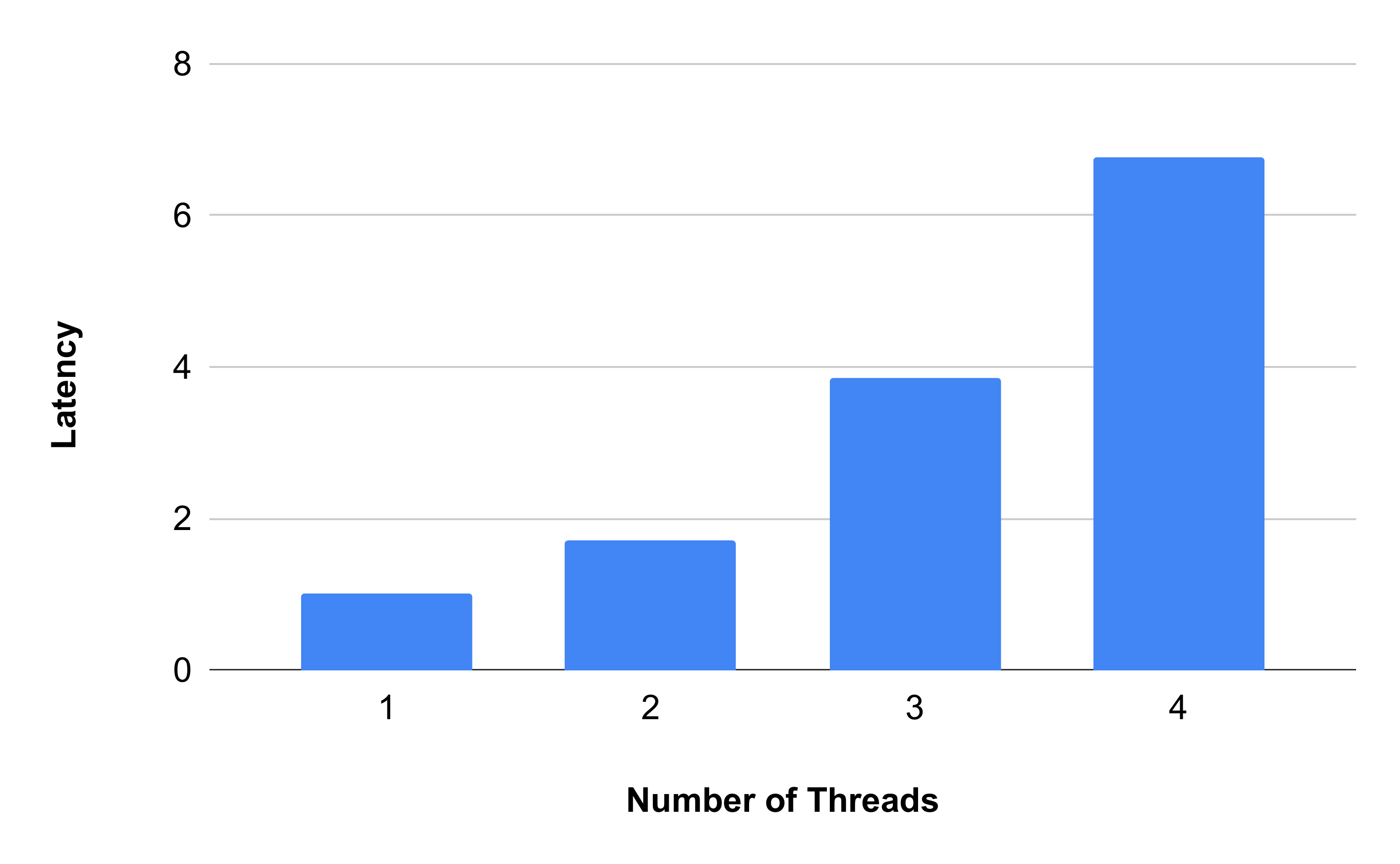}
 \vspace{-4mm}
 \caption{{Effect of Multi-threading on the execution time of training VGG16 relative to the single thread execution time}}
 \vspace{-4mm}
 \label{fig:thread}
\end{figure}

{\textbf{Parallelism in Baseline:} Using multiple SGX threads one can train concurrently on different batches. Figure~\ref{fig:thread} shows the effect of multi-threading. The latency grows with more threading within SGX. The main reason for this counter-intuitive behavior is that training large models are very memory intensive. But current implementations of SGX have limited memory encryption unit capacity. When the training data exceeds the memory limit there is substantial performance overhead in moving data between protected memory and untrusted DRAM. 
The same behaviour is observed in~\cite{tramer2018slalom}. So the the baseline with that multi-threading cannot provide a better performance for such memory intensive computations.}

\subsection{Inference Results}
\label{sec:infresult}
In this section we briefly explain Slalom~\cite{tramer2018slalom} a prior work that used SGX and GPU for inference and why that strategy is unable to handle training. Then we evaluate the speedup of DarKnight in the inference phase.

\textbf{Slalom Inference}
\label{sec:litreture}
Slalom is an inference framework that uses TEEs to protect data privacy and integrity. Slalom uses the Intel SGX enclave to blind input data $\mathbf x$ from a client with an additive stream cipher noise $\mathbf r$. The blinded data $(\mathbf x +\mathbf r)$ is then sent to an untrusted GPU where linear operations are performed. The computed data  $\mathbf W \cdot (\mathbf x +\mathbf r)$ is then returned to the enclave which can decode the correct computational output $\mathbf W \cdot \mathbf x$ by subtracting the precomputed $\mathbf W \cdot \mathbf r$. Here $\mathbf W$ is the model parameter matrix. 
But this encoding approach cannot be applied for training, since it precomputes $\mathbf W \cdot \mathbf r$. Precomputing the blinding factors is not feasible during training since the model parameters $\mathbf W$ are updated after processing every batch. Computing $\mathbf{W} \cdot \mathbf{r}$ inside the SGX after every batch also defeats the purpose of offloading the linear computations to GPU. {This fundamental challenge make it impossible to use Slalom encoding for training.} 
As such \name developed a more comprehensive encoding strategy  to handle training on private data {while providing rigorous privacy guarantee}. Nonetheless \name can perform inference on private data. Hence, we compare the inference performance of \name and Slalom.   
For a fair comparison, we implemented \name using Eigen library which is the design used in Slalom. Eigen is a high-performance C++ based linear algebra library. 

\textbf{Inference Speedup}: Fig.~\ref{fig:2} (a) compares the speedup of the inference for VGG16 and MobileNetV1 across five different configurations. The baseline bar (SGX) performs all the calculations within SGX. The red bar uses Slalom blinding while trusting GPU that results are always correct, DarKnight(4) is our model while using a virtual batch size of 4. Slalom+Integrity bar shows the performance when Slalom's data integrity verification (using Freivalds algorithm~\cite{blum1975toward}) is deployed to verify GPU computations. DarKnight(3)+Integrity uses DarKnight with a virtual batch size of 3 and an additional equation to redundantly compute all the results twice for integrity verification. 

For VGG16, DarKnight(4) provides $15X$ speedup, compared to the SGX only baseline, and $30\%$ improvement over Slalom. Slalom's implementation encrypts $\mathbf{W}\cdot\mathbf{r}$ and stores them outside of SGX memory. At each layer, they retrieve the necessary unblinding factors into SGX, then decrypt them before using them. This approach is required since their approach is required to compute the unblinding factors. \name does not need such a pre-computation. Instead it performs $1/K$ additional computations to decode the. Hence, \name performs additional computations on a GPU but reduces the SGX memory pressure. When providing the additional integrity checks, DarKnight(3) provides about $13X$ speedup over baseline and $1.45X$ speedup over Slalom. 

\textbf{Effect of Virtual Batch Size}: Recall that virtual batch size is the number of images that are linearly combined in the Equation~\eqref{eq:inference_blinding}. Fig.~\ref{fig:2}(b) quantifies the effect of batch size on the inference time. In the figure, DarKnight($K$) is used to denote a virtual batch size of $K$. For the same number of input data points with different batch sizes, we issue inference requests and divided the total inference time across four categories of operations: decoding, encoding, Relu, and Maxpooling operations. We used DarKnight(1) as the baseline. It represents the case where a single image is combined with a uniform random noise $r$ to create two equations using Equation~\eqref{eq:inference_blinding}. As the virtual batch size increases the total speedup improved as long as the virtual batch size fits within SGX memory limits. As the virtual batch size exceeds 4, the execution time gets worse due to SGX memory overflow in our current experimental setup. But as the SGX memory limitation is relaxed in future we believe large virtual batches can be processed to further improve \name's performance.

\vspace{-2mm}
\section{Conclusion}
\label{sec:con}
This work proposes DarKnight, a privacy and integrity preserving framework for DNNs' training and inference. DarKnight uses a hybrid execution model where TEE provide data obfuscation and untrusted GPUs provide computation acceleration. We design a data encoding in TEE for linear operations of DNN training and we provide a rigorous data privacy guarantee. DarKnight can also provide computation integrity and it is robust even in the presence of a malicious colluding GPUs. We evaluated different models and datasets and observe an average of $6.5X$ training speedup and $12.5X$ inference speedup without accuracy degradation over the baseline fully implemented inside TEE. 
\begin{acks}
We would like to express our special gratitude to Hsien-Hsin Sean Lee, Edward Suh, Wenjie Xiong, Chuan Guo and Mark Tygert for sharing their wisdom with us. We are immensely grateful to Krishna Giri Narra and Caroline Tripple for their valuable feedbacks on this project. We sincerely thank all the reviewers for their time and constructive comments. This material is based upon work supported by Defense Advanced Research Projects Agency (DARPA) under Contract Nos. HR001117C0053, HR001120C0088, NSF award number 5345039074, and Facebook AI Research Award Numbers 2215031173 and 2215031183. The views, opinions, and/or findings expressed are those of the author(s) and should not be interpreted as representing the official views or policies of the Department of Defense or the U.S. Government. 
\end{acks}
\newpage
\bibliographystyle{ACM-Reference-Format}
\bibliography{refs}


\begin{thebibliography}{88}


\ifx \showCODEN    \undefined \def \showCODEN     #1{\unskip}     \fi
\ifx \showDOI      \undefined \def \showDOI       #1{#1}\fi
\ifx \showISBNx    \undefined \def \showISBNx     #1{\unskip}     \fi
\ifx \showISBNxiii \undefined \def \showISBNxiii  #1{\unskip}     \fi
\ifx \showISSN     \undefined \def \showISSN      #1{\unskip}     \fi
\ifx \showLCCN     \undefined \def \showLCCN      #1{\unskip}     \fi
\ifx \shownote     \undefined \def \shownote      #1{#1}          \fi
\ifx \showarticletitle \undefined \def \showarticletitle #1{#1}   \fi
\ifx \showURL      \undefined \def \showURL       {\relax}        \fi
\providecommand\bibfield[2]{#2}
\providecommand\bibinfo[2]{#2}
\providecommand\natexlab[1]{#1}
\providecommand\showeprint[2][]{arXiv:#2}

\bibitem[\protect\citeauthoryear{Abadi, Chu, Goodfellow, McMahan, Mironov,
  Talwar, and Zhang}{Abadi et~al\mbox{.}}{2016}]%
        {abadi2016deep}
\bibfield{author}{\bibinfo{person}{Martin Abadi}, \bibinfo{person}{Andy Chu},
  \bibinfo{person}{Ian Goodfellow}, \bibinfo{person}{H~Brendan McMahan},
  \bibinfo{person}{Ilya Mironov}, \bibinfo{person}{Kunal Talwar}, {and}
  \bibinfo{person}{Li Zhang}.} \bibinfo{year}{2016}\natexlab{}.
\newblock \showarticletitle{Deep learning with differential privacy}. In
  \bibinfo{booktitle}{\emph{Proceedings of the 2016 ACM SIGSAC Conference on
  Computer and Communications Security}}. \bibinfo{pages}{308--318}.
\newblock


\bibitem[\protect\citeauthoryear{Alves}{Alves}{2004}]%
        {alves2004trustzone}
\bibfield{author}{\bibinfo{person}{Tiago Alves}.}
  \bibinfo{year}{2004}\natexlab{}.
\newblock \showarticletitle{Trustzone: Integrated hardware and software
  security}.
\newblock \bibinfo{journal}{\emph{White paper}} (\bibinfo{year}{2004}).
\newblock


\bibitem[\protect\citeauthoryear{Amazon}{Amazon}{2020}]%
        {Amazon}
\bibfield{author}{\bibinfo{person}{Amazon}.} \bibinfo{year}{2020}\natexlab{}.
\newblock \bibinfo{booktitle}{\emph{Machine Learning on AWS}}.
\newblock
\urldef\tempurl%
\url{https://aws.amazon.com/machine-learning}
\showURL{%
\tempurl}


\bibitem[\protect\citeauthoryear{Asvadishirehjini, Kantarcioglu, and
  Malin}{Asvadishirehjini et~al\mbox{.}}{2020}]%
        {asvadishirehjini2020goat}
\bibfield{author}{\bibinfo{person}{Aref Asvadishirehjini},
  \bibinfo{person}{Murat Kantarcioglu}, {and} \bibinfo{person}{Bradley Malin}.}
  \bibinfo{year}{2020}\natexlab{}.
\newblock \showarticletitle{GOAT: GPU Outsourcing of Deep Learning Training
  With Asynchronous Probabilistic Integrity Verification Inside Trusted
  Execution Environment}.
\newblock \bibinfo{journal}{\emph{arXiv preprint arXiv:2010.08855}}
  (\bibinfo{year}{2020}).
\newblock


\bibitem[\protect\citeauthoryear{Blatt, Gusev, Polyakov, and Goldwasser}{Blatt
  et~al\mbox{.}}{2020}]%
        {blatt2020secure}
\bibfield{author}{\bibinfo{person}{Marcelo Blatt}, \bibinfo{person}{Alexander
  Gusev}, \bibinfo{person}{Yuriy Polyakov}, {and} \bibinfo{person}{Shafi
  Goldwasser}.} \bibinfo{year}{2020}\natexlab{}.
\newblock \showarticletitle{Secure large-scale genome-wide association studies
  using homomorphic encryption}.
\newblock \bibinfo{journal}{\emph{Proceedings of the National Academy of
  Sciences}} \bibinfo{volume}{117}, \bibinfo{number}{21}
  (\bibinfo{year}{2020}), \bibinfo{pages}{11608--11613}.
\newblock


\bibitem[\protect\citeauthoryear{Blum and Blum}{Blum and Blum}{1975}]%
        {blum1975toward}
\bibfield{author}{\bibinfo{person}{Lenore Blum} {and} \bibinfo{person}{Manuel
  Blum}.} \bibinfo{year}{1975}\natexlab{}.
\newblock \showarticletitle{Toward a mathematical theory of inductive
  inference}.
\newblock \bibinfo{journal}{\emph{Information and control}}
  \bibinfo{volume}{28}, \bibinfo{number}{2} (\bibinfo{year}{1975}),
  \bibinfo{pages}{125--155}.
\newblock


\bibitem[\protect\citeauthoryear{Bonawitz, Ivanov, Kreuter, Marcedone, McMahan,
  Patel, Ramage, Segal, and Seth}{Bonawitz et~al\mbox{.}}{2017}]%
        {bonawitz2017practical}
\bibfield{author}{\bibinfo{person}{Keith Bonawitz}, \bibinfo{person}{Vladimir
  Ivanov}, \bibinfo{person}{Ben Kreuter}, \bibinfo{person}{Antonio Marcedone},
  \bibinfo{person}{H~Brendan McMahan}, \bibinfo{person}{Sarvar Patel},
  \bibinfo{person}{Daniel Ramage}, \bibinfo{person}{Aaron Segal}, {and}
  \bibinfo{person}{Karn Seth}.} \bibinfo{year}{2017}\natexlab{}.
\newblock \showarticletitle{Practical secure aggregation for privacy-preserving
  machine learning}. In \bibinfo{booktitle}{\emph{Proceedings of the 2017 ACM
  SIGSAC Conference on Computer and Communications Security}}.
  \bibinfo{pages}{1175--1191}.
\newblock


\bibitem[\protect\citeauthoryear{Brasser, M{\"u}ller, Dmitrienko, Kostiainen,
  Capkun, and Sadeghi}{Brasser et~al\mbox{.}}{2017}]%
        {brasser2017software}
\bibfield{author}{\bibinfo{person}{Ferdinand Brasser}, \bibinfo{person}{Urs
  M{\"u}ller}, \bibinfo{person}{Alexandra Dmitrienko}, \bibinfo{person}{Kari
  Kostiainen}, \bibinfo{person}{Srdjan Capkun}, {and}
  \bibinfo{person}{Ahmad-Reza Sadeghi}.} \bibinfo{year}{2017}\natexlab{}.
\newblock \showarticletitle{Software grand exposure:$\{$SGX$\}$ cache attacks
  are practical}. In \bibinfo{booktitle}{\emph{11th $\{$USENIX$\}$ Workshop on
  Offensive Technologies ($\{$WOOT$\}$ 17)}}.
\newblock


\bibitem[\protect\citeauthoryear{Canziani, Paszke, and Culurciello}{Canziani
  et~al\mbox{.}}{2016}]%
        {canziani2016analysis}
\bibfield{author}{\bibinfo{person}{Alfredo Canziani}, \bibinfo{person}{Adam
  Paszke}, {and} \bibinfo{person}{Eugenio Culurciello}.}
  \bibinfo{year}{2016}\natexlab{}.
\newblock \showarticletitle{An analysis of deep neural network models for
  practical applications}.
\newblock \bibinfo{journal}{\emph{arXiv preprint arXiv:1605.07678}}
  (\bibinfo{year}{2016}).
\newblock


\bibitem[\protect\citeauthoryear{Carlini, Deng, Garg, Jha, Mahloujifar,
  Mahmoody, Song, Thakurta, and Tramer}{Carlini et~al\mbox{.}}{2020}]%
        {carlini2020attack}
\bibfield{author}{\bibinfo{person}{Nicholas Carlini}, \bibinfo{person}{Samuel
  Deng}, \bibinfo{person}{Sanjam Garg}, \bibinfo{person}{Somesh Jha},
  \bibinfo{person}{Saeed Mahloujifar}, \bibinfo{person}{Mohammad Mahmoody},
  \bibinfo{person}{Shuang Song}, \bibinfo{person}{Abhradeep Thakurta}, {and}
  \bibinfo{person}{Florian Tramer}.} \bibinfo{year}{2020}\natexlab{}.
\newblock \showarticletitle{An Attack on InstaHide: Is Private Learning
  Possible with Instance Encoding?}
\newblock \bibinfo{journal}{\emph{arXiv preprint arXiv:2011.05315}}
  (\bibinfo{year}{2020}).
\newblock


\bibitem[\protect\citeauthoryear{Costan and Devadas}{Costan and
  Devadas}{2016}]%
        {costan2016intel}
\bibfield{author}{\bibinfo{person}{Victor Costan} {and}
  \bibinfo{person}{Srinivas Devadas}.} \bibinfo{year}{2016}\natexlab{}.
\newblock \showarticletitle{Intel SGX Explained.}
\newblock \bibinfo{journal}{\emph{IACR Cryptology ePrint Archive}}
  \bibinfo{volume}{2016}, \bibinfo{number}{086} (\bibinfo{year}{2016}),
  \bibinfo{pages}{1--118}.
\newblock


\bibitem[\protect\citeauthoryear{Costan, Lebedev, and Devadas}{Costan
  et~al\mbox{.}}{2016}]%
        {costan2016sanctum}
\bibfield{author}{\bibinfo{person}{Victor Costan}, \bibinfo{person}{Ilia
  Lebedev}, {and} \bibinfo{person}{Srinivas Devadas}.}
  \bibinfo{year}{2016}\natexlab{}.
\newblock \showarticletitle{Sanctum: Minimal hardware extensions for strong
  software isolation}. In \bibinfo{booktitle}{\emph{25th $\{$USENIX$\}$
  Security Symposium ($\{$USENIX$\}$ Security 16)}}. \bibinfo{pages}{857--874}.
\newblock


\bibitem[\protect\citeauthoryear{Cover}{Cover}{1999}]%
        {cover1999elements}
\bibfield{author}{\bibinfo{person}{Thomas~M Cover}.}
  \bibinfo{year}{1999}\natexlab{}.
\newblock \bibinfo{booktitle}{\emph{Elements of information theory}}.
\newblock \bibinfo{publisher}{John Wiley \& Sons}.
\newblock


\bibitem[\protect\citeauthoryear{Cox}{Cox}{1994}]%
        {cox1994matrix}
\bibfield{author}{\bibinfo{person}{LH Cox}.} \bibinfo{year}{1994}\natexlab{}.
\newblock \showarticletitle{Matrix masking methods for disclosure limitation in
  microdata}.
\newblock \bibinfo{journal}{\emph{Surv. Methodol.}}  \bibinfo{volume}{20}
  (\bibinfo{year}{1994}), \bibinfo{pages}{165--169}.
\newblock


\bibitem[\protect\citeauthoryear{Cox}{Cox}{1980}]%
        {cox1980suppression}
\bibfield{author}{\bibinfo{person}{Lawrence~H Cox}.}
  \bibinfo{year}{1980}\natexlab{}.
\newblock \showarticletitle{Suppression methodology and statistical disclosure
  control}.
\newblock \bibinfo{journal}{\emph{J. Amer. Statist. Assoc.}}
  \bibinfo{volume}{75}, \bibinfo{number}{370} (\bibinfo{year}{1980}),
  \bibinfo{pages}{377--385}.
\newblock


\bibitem[\protect\citeauthoryear{Deng and Long}{Deng and Long}{2004}]%
        {deng2004secure}
\bibfield{author}{\bibinfo{person}{Fu-Guo Deng} {and} \bibinfo{person}{Gui~Lu
  Long}.} \bibinfo{year}{2004}\natexlab{}.
\newblock \showarticletitle{Secure direct communication with a quantum one-time
  pad}.
\newblock \bibinfo{journal}{\emph{Physical Review A}} \bibinfo{volume}{69},
  \bibinfo{number}{5} (\bibinfo{year}{2004}), \bibinfo{pages}{052319}.
\newblock


\bibitem[\protect\citeauthoryear{Dong, Cheng, Hossain, and Leung}{Dong
  et~al\mbox{.}}{2019}]%
        {dong2019secure}
\bibfield{author}{\bibinfo{person}{Yanjie Dong}, \bibinfo{person}{Julian
  Cheng}, \bibinfo{person}{Md~Jahangir Hossain}, {and}
  \bibinfo{person}{Victor~CM Leung}.} \bibinfo{year}{2019}\natexlab{}.
\newblock \showarticletitle{Secure distributed on-device learning networks with
  byzantine adversaries}.
\newblock \bibinfo{journal}{\emph{IEEE Network}} \bibinfo{volume}{33},
  \bibinfo{number}{6} (\bibinfo{year}{2019}), \bibinfo{pages}{180--187}.
\newblock


\bibitem[\protect\citeauthoryear{Erlingsson, Feldman, Mironov, Raghunathan,
  Talwar, and Thakurta}{Erlingsson et~al\mbox{.}}{2019}]%
        {erlingsson2019amplification}
\bibfield{author}{\bibinfo{person}{{\'U}lfar Erlingsson},
  \bibinfo{person}{Vitaly Feldman}, \bibinfo{person}{Ilya Mironov},
  \bibinfo{person}{Ananth Raghunathan}, \bibinfo{person}{Kunal Talwar}, {and}
  \bibinfo{person}{Abhradeep Thakurta}.} \bibinfo{year}{2019}\natexlab{}.
\newblock \showarticletitle{Amplification by shuffling: From local to central
  differential privacy via anonymity}. In \bibinfo{booktitle}{\emph{Proceedings
  of the Thirtieth Annual ACM-SIAM Symposium on Discrete Algorithms}}. SIAM,
  \bibinfo{pages}{2468--2479}.
\newblock


\bibitem[\protect\citeauthoryear{Erlingsson, Pihur, and Korolova}{Erlingsson
  et~al\mbox{.}}{2014}]%
        {erlingsson2014rappor}
\bibfield{author}{\bibinfo{person}{{\'U}lfar Erlingsson},
  \bibinfo{person}{Vasyl Pihur}, {and} \bibinfo{person}{Aleksandra Korolova}.}
  \bibinfo{year}{2014}\natexlab{}.
\newblock \showarticletitle{Rappor: Randomized aggregatable privacy-preserving
  ordinal response}. In \bibinfo{booktitle}{\emph{Proceedings of the 2014 ACM
  SIGSAC conference on computer and communications security}}.
  \bibinfo{pages}{1054--1067}.
\newblock


\bibitem[\protect\citeauthoryear{Feng, Hashemi, Annavaram, and Narayanan}{Feng
  et~al\mbox{.}}{2022}]%
        {feng2022enhancing}
\bibfield{author}{\bibinfo{person}{Tiantian Feng}, \bibinfo{person}{Hanieh
  Hashemi}, \bibinfo{person}{Murali Annavaram}, {and}
  \bibinfo{person}{Shrikanth~S Narayanan}.} \bibinfo{year}{2022}\natexlab{}.
\newblock \showarticletitle{Enhancing Privacy Through Domain Adaptive Noise
  Injection For Speech Emotion Recognition}. In
  \bibinfo{booktitle}{\emph{ICASSP 2022-2022 IEEE International Conference on
  Acoustics, Speech and Signal Processing (ICASSP)}}. IEEE,
  \bibinfo{pages}{7702--7706}.
\newblock


\bibitem[\protect\citeauthoryear{Foerster, Assael, De~Freitas, and
  Whiteson}{Foerster et~al\mbox{.}}{2016}]%
        {foerster2016learning}
\bibfield{author}{\bibinfo{person}{Jakob Foerster},
  \bibinfo{person}{Ioannis~Alexandros Assael}, \bibinfo{person}{Nando
  De~Freitas}, {and} \bibinfo{person}{Shimon Whiteson}.}
  \bibinfo{year}{2016}\natexlab{}.
\newblock \showarticletitle{Learning to communicate with deep multi-agent
  reinforcement learning}. In \bibinfo{booktitle}{\emph{Advances in neural
  information processing systems}}. \bibinfo{pages}{2137--2145}.
\newblock


\bibitem[\protect\citeauthoryear{Geiping, Bauermeister, Dr{\"o}ge, and
  Moeller}{Geiping et~al\mbox{.}}{2020}]%
        {geiping2020inverting}
\bibfield{author}{\bibinfo{person}{Jonas Geiping}, \bibinfo{person}{Hartmut
  Bauermeister}, \bibinfo{person}{Hannah Dr{\"o}ge}, {and}
  \bibinfo{person}{Michael Moeller}.} \bibinfo{year}{2020}\natexlab{}.
\newblock \showarticletitle{Inverting Gradients--How easy is it to break
  privacy in federated learning?}
\newblock \bibinfo{journal}{\emph{arXiv preprint arXiv:2003.14053}}
  (\bibinfo{year}{2020}).
\newblock


\bibitem[\protect\citeauthoryear{Gentry}{Gentry}{2009}]%
        {gentry2009fully}
\bibfield{author}{\bibinfo{person}{Craig Gentry}.}
  \bibinfo{year}{2009}\natexlab{}.
\newblock \showarticletitle{Fully homomorphic encryption using ideal lattices}.
  In \bibinfo{booktitle}{\emph{Proceedings of the forty-first annual ACM
  symposium on Theory of computing}}. \bibinfo{pages}{169--178}.
\newblock


\bibitem[\protect\citeauthoryear{Ghodsi, Gu, and Garg}{Ghodsi
  et~al\mbox{.}}{2017}]%
        {ghodsi2017safetynets}
\bibfield{author}{\bibinfo{person}{Zahra Ghodsi}, \bibinfo{person}{Tianyu Gu},
  {and} \bibinfo{person}{Siddharth Garg}.} \bibinfo{year}{2017}\natexlab{}.
\newblock \showarticletitle{Safetynets: Verifiable execution of deep neural
  networks on an untrusted cloud}.
\newblock \bibinfo{journal}{\emph{arXiv preprint arXiv:1706.10268}}
  (\bibinfo{year}{2017}).
\newblock


\bibitem[\protect\citeauthoryear{Gilad-Bachrach, Dowlin, Laine, Lauter,
  Naehrig, and Wernsing}{Gilad-Bachrach et~al\mbox{.}}{2016}]%
        {gilad2016cryptonets}
\bibfield{author}{\bibinfo{person}{Ran Gilad-Bachrach}, \bibinfo{person}{Nathan
  Dowlin}, \bibinfo{person}{Kim Laine}, \bibinfo{person}{Kristin Lauter},
  \bibinfo{person}{Michael Naehrig}, {and} \bibinfo{person}{John Wernsing}.}
  \bibinfo{year}{2016}\natexlab{}.
\newblock \showarticletitle{Cryptonets: Applying neural networks to encrypted
  data with high throughput and accuracy}. In
  \bibinfo{booktitle}{\emph{International Conference on Machine Learning}}.
  \bibinfo{pages}{201--210}.
\newblock


\bibitem[\protect\citeauthoryear{Goldreich}{Goldreich}{2007}]%
        {goldreich2007foundations}
\bibfield{author}{\bibinfo{person}{Oded Goldreich}.}
  \bibinfo{year}{2007}\natexlab{}.
\newblock \bibinfo{booktitle}{\emph{Foundations of cryptography: volume 1,
  basic tools}}.
\newblock \bibinfo{publisher}{Cambridge university press}.
\newblock


\bibitem[\protect\citeauthoryear{Google}{Google}{2020}]%
        {Google}
\bibfield{author}{\bibinfo{person}{Google}.} \bibinfo{year}{2020}\natexlab{}.
\newblock \bibinfo{booktitle}{\emph{Google AI platform}}.
\newblock
\urldef\tempurl%
\url{https://cloud.google.com/products/ai}
\showURL{%
\tempurl}


\bibitem[\protect\citeauthoryear{G{\"o}tzfried, Eckert, Schinzel, and
  M{\"u}ller}{G{\"o}tzfried et~al\mbox{.}}{2017}]%
        {gotzfried2017cache}
\bibfield{author}{\bibinfo{person}{Johannes G{\"o}tzfried},
  \bibinfo{person}{Moritz Eckert}, \bibinfo{person}{Sebastian Schinzel}, {and}
  \bibinfo{person}{Tilo M{\"u}ller}.} \bibinfo{year}{2017}\natexlab{}.
\newblock \showarticletitle{Cache attacks on Intel SGX}. In
  \bibinfo{booktitle}{\emph{Proceedings of the 10th European Workshop on
  Systems Security}}. \bibinfo{pages}{1--6}.
\newblock


\bibitem[\protect\citeauthoryear{Grubbs, Khandelwal, Lacharit{\'e}, Brown, Li,
  Agarwal, and Ristenpart}{Grubbs et~al\mbox{.}}{2020}]%
        {pancakge2020}
\bibfield{author}{\bibinfo{person}{Paul Grubbs}, \bibinfo{person}{Anurag
  Khandelwal}, \bibinfo{person}{Marie-Sarah Lacharit{\'e}},
  \bibinfo{person}{Lloyd Brown}, \bibinfo{person}{Lucy Li},
  \bibinfo{person}{Rachit Agarwal}, {and} \bibinfo{person}{Thomas Ristenpart}.}
  \bibinfo{year}{2020}\natexlab{}.
\newblock \showarticletitle{Pancake: Frequency Smoothing for Encrypted Data
  Stores}. In \bibinfo{booktitle}{\emph{29th USENIX Security Symposium (USENIX
  Security 20)}}. \bibinfo{publisher}{USENIX Association},
  \bibinfo{pages}{2451--2468}.
\newblock
\showISBNx{978-1-939133-17-5}
\urldef\tempurl%
\url{https://www.usenix.org/conference/usenixsecurity20/presentation/grubbs}
\showURL{%
\tempurl}


\bibitem[\protect\citeauthoryear{Gupta, Agrawal, Gopalakrishnan, and
  Narayanan}{Gupta et~al\mbox{.}}{2015}]%
        {gupta2015deep}
\bibfield{author}{\bibinfo{person}{Suyog Gupta}, \bibinfo{person}{Ankur
  Agrawal}, \bibinfo{person}{Kailash Gopalakrishnan}, {and}
  \bibinfo{person}{Pritish Narayanan}.} \bibinfo{year}{2015}\natexlab{}.
\newblock \showarticletitle{Deep learning with limited numerical precision}. In
  \bibinfo{booktitle}{\emph{International Conference on Machine Learning}}.
  \bibinfo{pages}{1737--1746}.
\newblock


\bibitem[\protect\citeauthoryear{Han, Mao, and Dally}{Han
  et~al\mbox{.}}{2015}]%
        {han2015deep}
\bibfield{author}{\bibinfo{person}{Song Han}, \bibinfo{person}{Huizi Mao},
  {and} \bibinfo{person}{William~J Dally}.} \bibinfo{year}{2015}\natexlab{}.
\newblock \showarticletitle{Deep compression: Compressing deep neural networks
  with pruning, trained quantization and huffman coding}.
\newblock \bibinfo{journal}{\emph{arXiv preprint arXiv:1510.00149}}
  (\bibinfo{year}{2015}).
\newblock


\bibitem[\protect\citeauthoryear{Hashemi, Wang, Guo, and Annavaram}{Hashemi
  et~al\mbox{.}}{2021}]%
        {hashemi2021byzantine}
\bibfield{author}{\bibinfo{person}{Hanieh Hashemi}, \bibinfo{person}{Yongqin
  Wang}, \bibinfo{person}{Chuan Guo}, {and} \bibinfo{person}{Murali
  Annavaram}.} \bibinfo{year}{2021}\natexlab{}.
\newblock \showarticletitle{Byzantine-Robust and Privacy-Preserving Framework
  for FedML}.
\newblock \bibinfo{journal}{\emph{arXiv preprint arXiv:2105.02295}}
  (\bibinfo{year}{2021}).
\newblock


\bibitem[\protect\citeauthoryear{He, Zhang, Ren, and Sun}{He
  et~al\mbox{.}}{2016}]%
        {he2016deep}
\bibfield{author}{\bibinfo{person}{Kaiming He}, \bibinfo{person}{Xiangyu
  Zhang}, \bibinfo{person}{Shaoqing Ren}, {and} \bibinfo{person}{Jian Sun}.}
  \bibinfo{year}{2016}\natexlab{}.
\newblock \showarticletitle{Deep residual learning for image recognition}. In
  \bibinfo{booktitle}{\emph{Proceedings of the IEEE conference on computer
  vision and pattern recognition}}. \bibinfo{pages}{770--778}.
\newblock


\bibitem[\protect\citeauthoryear{Heaton, Polson, and Witte}{Heaton
  et~al\mbox{.}}{2017}]%
        {heaton2017deep}
\bibfield{author}{\bibinfo{person}{JB Heaton}, \bibinfo{person}{NG Polson},
  {and} \bibinfo{person}{Jan~Hendrik Witte}.} \bibinfo{year}{2017}\natexlab{}.
\newblock \showarticletitle{Deep learning for finance: deep portfolios}.
\newblock \bibinfo{journal}{\emph{Applied Stochastic Models in Business and
  Industry}} \bibinfo{volume}{33}, \bibinfo{number}{1} (\bibinfo{year}{2017}),
  \bibinfo{pages}{3--12}.
\newblock


\bibitem[\protect\citeauthoryear{Hua, Umar, Zhang, and Suh}{Hua
  et~al\mbox{.}}{2020}]%
        {hua2020guardnn}
\bibfield{author}{\bibinfo{person}{Weizhe Hua}, \bibinfo{person}{Muhammad
  Umar}, \bibinfo{person}{Zhiru Zhang}, {and} \bibinfo{person}{G~Edward Suh}.}
  \bibinfo{year}{2020}\natexlab{}.
\newblock \showarticletitle{GuardNN: Secure DNN Accelerator for
  Privacy-Preserving Deep Learning}.
\newblock \bibinfo{journal}{\emph{arXiv preprint arXiv:2008.11632}}
  (\bibinfo{year}{2020}).
\newblock


\bibitem[\protect\citeauthoryear{Huang, Song, Li, and Arora}{Huang
  et~al\mbox{.}}{2020}]%
        {huang2020instahide}
\bibfield{author}{\bibinfo{person}{Yangsibo Huang}, \bibinfo{person}{Zhao
  Song}, \bibinfo{person}{Kai Li}, {and} \bibinfo{person}{Sanjeev Arora}.}
  \bibinfo{year}{2020}\natexlab{}.
\newblock \showarticletitle{Instahide: Instance-hiding schemes for private
  distributed learning}. In \bibinfo{booktitle}{\emph{International Conference
  on Machine Learning}}. PMLR, \bibinfo{pages}{4507--4518}.
\newblock


\bibitem[\protect\citeauthoryear{Hunt, Song, Shokri, Shmatikov, and
  Witchel}{Hunt et~al\mbox{.}}{2018}]%
        {hunt2018chiron}
\bibfield{author}{\bibinfo{person}{Tyler Hunt}, \bibinfo{person}{Congzheng
  Song}, \bibinfo{person}{Reza Shokri}, \bibinfo{person}{Vitaly Shmatikov},
  {and} \bibinfo{person}{Emmett Witchel}.} \bibinfo{year}{2018}\natexlab{}.
\newblock \showarticletitle{Chiron: Privacy-preserving machine learning as a
  service}.
\newblock \bibinfo{journal}{\emph{arXiv preprint arXiv:1803.05961}}
  (\bibinfo{year}{2018}).
\newblock


\bibitem[\protect\citeauthoryear{Hynes, Cheng, and Song}{Hynes
  et~al\mbox{.}}{2018}]%
        {hynes2018efficient}
\bibfield{author}{\bibinfo{person}{Nick Hynes}, \bibinfo{person}{Raymond
  Cheng}, {and} \bibinfo{person}{Dawn Song}.} \bibinfo{year}{2018}\natexlab{}.
\newblock \showarticletitle{Efficient deep learning on multi-source private
  data}.
\newblock \bibinfo{journal}{\emph{arXiv preprint arXiv:1807.06689}}
  (\bibinfo{year}{2018}).
\newblock


\bibitem[\protect\citeauthoryear{Juvekar, Vaikuntanathan, and
  Chandrakasan}{Juvekar et~al\mbox{.}}{2018}]%
        {juvekar2018gazelle}
\bibfield{author}{\bibinfo{person}{Chiraag Juvekar}, \bibinfo{person}{Vinod
  Vaikuntanathan}, {and} \bibinfo{person}{Anantha Chandrakasan}.}
  \bibinfo{year}{2018}\natexlab{}.
\newblock \showarticletitle{$\{$GAZELLE$\}$: A low latency framework for secure
  neural network inference}. In \bibinfo{booktitle}{\emph{27th $\{$USENIX$\}$
  Security Symposium ($\{$USENIX$\}$ Security 18)}}.
  \bibinfo{pages}{1651--1669}.
\newblock


\bibitem[\protect\citeauthoryear{Kim}{Kim}{1986}]%
        {kim1986method}
\bibfield{author}{\bibinfo{person}{Jay~J Kim}.}
  \bibinfo{year}{1986}\natexlab{}.
\newblock \showarticletitle{A method for limiting disclosure in microdata based
  on random noise and transformation}. In \bibinfo{booktitle}{\emph{Proceedings
  of the section on survey research methods}}. American Statistical Association
  Alexandria, VA, \bibinfo{pages}{303--308}.
\newblock


\bibitem[\protect\citeauthoryear{Krizhevsky, Hinton, et~al\mbox{.}}{Krizhevsky
  et~al\mbox{.}}{2009}]%
        {krizhevsky2009learning}
\bibfield{author}{\bibinfo{person}{Alex Krizhevsky}, \bibinfo{person}{Geoffrey
  Hinton}, {et~al\mbox{.}}} \bibinfo{year}{2009}\natexlab{}.
\newblock \showarticletitle{Learning multiple layers of features from tiny
  images}.
\newblock \bibinfo{journal}{\emph{online: http://www. cs. toronto.
  edu/kriz/cifar. html}} (\bibinfo{year}{2009}).
\newblock


\bibitem[\protect\citeauthoryear{Lee, Lin, Pushp, Li, Liu, Lee, Xu, Xu, Zhang,
  and Song}{Lee et~al\mbox{.}}{2019}]%
        {lee2019occlumency}
\bibfield{author}{\bibinfo{person}{Taegyeong Lee}, \bibinfo{person}{Zhiqi Lin},
  \bibinfo{person}{Saumay Pushp}, \bibinfo{person}{Caihua Li},
  \bibinfo{person}{Yunxin Liu}, \bibinfo{person}{Youngki Lee},
  \bibinfo{person}{Fengyuan Xu}, \bibinfo{person}{Chenren Xu},
  \bibinfo{person}{Lintao Zhang}, {and} \bibinfo{person}{Junehwa Song}.}
  \bibinfo{year}{2019}\natexlab{}.
\newblock \showarticletitle{Occlumency: Privacy-preserving remote deep-learning
  inference using SGX}. In \bibinfo{booktitle}{\emph{The 25th Annual
  International Conference on Mobile Computing and Networking}}.
  \bibinfo{pages}{1--17}.
\newblock


\bibitem[\protect\citeauthoryear{Lin, Talathi, and Annapureddy}{Lin
  et~al\mbox{.}}{2016}]%
        {lin2016fixed}
\bibfield{author}{\bibinfo{person}{Darryl Lin}, \bibinfo{person}{Sachin
  Talathi}, {and} \bibinfo{person}{Sreekanth Annapureddy}.}
  \bibinfo{year}{2016}\natexlab{}.
\newblock \showarticletitle{Fixed point quantization of deep convolutional
  networks}. In \bibinfo{booktitle}{\emph{International conference on machine
  learning}}. PMLR, \bibinfo{pages}{2849--2858}.
\newblock


\bibitem[\protect\citeauthoryear{Lin, Gan, and Han}{Lin et~al\mbox{.}}{2019}]%
        {lin2019defensive}
\bibfield{author}{\bibinfo{person}{Ji Lin}, \bibinfo{person}{Chuang Gan}, {and}
  \bibinfo{person}{Song Han}.} \bibinfo{year}{2019}\natexlab{}.
\newblock \showarticletitle{Defensive quantization: When efficiency meets
  robustness}.
\newblock \bibinfo{journal}{\emph{arXiv preprint arXiv:1904.08444}}
  (\bibinfo{year}{2019}).
\newblock


\bibitem[\protect\citeauthoryear{Liu, Juuti, Lu, and Asokan}{Liu
  et~al\mbox{.}}{2017}]%
        {liu2017oblivious}
\bibfield{author}{\bibinfo{person}{Jian Liu}, \bibinfo{person}{Mika Juuti},
  \bibinfo{person}{Yao Lu}, {and} \bibinfo{person}{Nadarajah Asokan}.}
  \bibinfo{year}{2017}\natexlab{}.
\newblock \showarticletitle{Oblivious neural network predictions via minionn
  transformations}. In \bibinfo{booktitle}{\emph{Proceedings of the 2017 ACM
  SIGSAC Conference on Computer and Communications Security}}.
  \bibinfo{pages}{619--631}.
\newblock


\bibitem[\protect\citeauthoryear{Microsoft}{Microsoft}{2020}]%
        {AzureML}
\bibfield{author}{\bibinfo{person}{Microsoft}.}
  \bibinfo{year}{2020}\natexlab{}.
\newblock \bibinfo{booktitle}{\emph{Azure Machine Learning}}.
\newblock
\urldef\tempurl%
\url{https://azure.microsoft.com/en-us/services/machine-learning}
\showURL{%
\tempurl}


\bibitem[\protect\citeauthoryear{Mireshghallah, Taram, Jalali, Elthakeb,
  Tullsen, and Esmaeilzadeh}{Mireshghallah et~al\mbox{.}}{2021}]%
        {mireshghallah2021not}
\bibfield{author}{\bibinfo{person}{Fatemehsadat Mireshghallah},
  \bibinfo{person}{Mohammadkazem Taram}, \bibinfo{person}{Ali Jalali},
  \bibinfo{person}{Ahmed Taha~Taha Elthakeb}, \bibinfo{person}{Dean Tullsen},
  {and} \bibinfo{person}{Hadi Esmaeilzadeh}.} \bibinfo{year}{2021}\natexlab{}.
\newblock \showarticletitle{Not all features are equal: Discovering essential
  features for preserving prediction privacy}. In
  \bibinfo{booktitle}{\emph{Proceedings of the Web Conference 2021}}.
  \bibinfo{pages}{669--680}.
\newblock


\bibitem[\protect\citeauthoryear{Mireshghallah, Taram, Ramrakhyani, Jalali,
  Tullsen, and Esmaeilzadeh}{Mireshghallah et~al\mbox{.}}{2020}]%
        {mireshghallah2020shredder}
\bibfield{author}{\bibinfo{person}{Fatemehsadat Mireshghallah},
  \bibinfo{person}{Mohammadkazem Taram}, \bibinfo{person}{Prakash Ramrakhyani},
  \bibinfo{person}{Ali Jalali}, \bibinfo{person}{Dean Tullsen}, {and}
  \bibinfo{person}{Hadi Esmaeilzadeh}.} \bibinfo{year}{2020}\natexlab{}.
\newblock \showarticletitle{Shredder: Learning noise distributions to protect
  inference privacy}. In \bibinfo{booktitle}{\emph{Proceedings of the
  Twenty-Fifth International Conference on Architectural Support for
  Programming Languages and Operating Systems}}. \bibinfo{pages}{3--18}.
\newblock


\bibitem[\protect\citeauthoryear{Mirshghallah, Taram, Vepakomma, Singh, Raskar,
  and Esmaeilzadeh}{Mirshghallah et~al\mbox{.}}{2020}]%
        {mirshghallah2020privacy}
\bibfield{author}{\bibinfo{person}{Fatemehsadat Mirshghallah},
  \bibinfo{person}{Mohammadkazem Taram}, \bibinfo{person}{Praneeth Vepakomma},
  \bibinfo{person}{Abhishek Singh}, \bibinfo{person}{Ramesh Raskar}, {and}
  \bibinfo{person}{Hadi Esmaeilzadeh}.} \bibinfo{year}{2020}\natexlab{}.
\newblock \showarticletitle{Privacy in deep learning: A survey}.
\newblock \bibinfo{journal}{\emph{arXiv preprint arXiv:2004.12254}}
  (\bibinfo{year}{2020}).
\newblock


\bibitem[\protect\citeauthoryear{Mishra, Lehmkuhl, Srinivasan, Zheng, and
  Popa}{Mishra et~al\mbox{.}}{2020}]%
        {mishra2020delphi}
\bibfield{author}{\bibinfo{person}{Pratyush Mishra}, \bibinfo{person}{Ryan
  Lehmkuhl}, \bibinfo{person}{Akshayaram Srinivasan}, \bibinfo{person}{Wenting
  Zheng}, {and} \bibinfo{person}{Raluca~Ada Popa}.}
  \bibinfo{year}{2020}\natexlab{}.
\newblock \showarticletitle{DELPHI: A cryptographic inference service for
  neural networks}. In \bibinfo{booktitle}{\emph{29th $\{$USENIX$\}$ Security
  Symposium ($\{$USENIX$\}$ Security 20)}}.
\newblock


\bibitem[\protect\citeauthoryear{Mo, Haddadi, Katevas, Marin, Perino, and
  Kourtellis}{Mo et~al\mbox{.}}{2021}]%
        {mo2021ppfl}
\bibfield{author}{\bibinfo{person}{Fan Mo}, \bibinfo{person}{Hamed Haddadi},
  \bibinfo{person}{Kleomenis Katevas}, \bibinfo{person}{Eduard Marin},
  \bibinfo{person}{Diego Perino}, {and} \bibinfo{person}{Nicolas Kourtellis}.}
  \bibinfo{year}{2021}\natexlab{}.
\newblock \showarticletitle{PPFL: privacy-preserving federated learning with
  trusted execution environments}.
\newblock \bibinfo{journal}{\emph{arXiv preprint arXiv:2104.14380}}
  (\bibinfo{year}{2021}).
\newblock


\bibitem[\protect\citeauthoryear{Mo, Shamsabadi, Katevas, Demetriou,
  Leontiadis, Cavallaro, and Haddadi}{Mo et~al\mbox{.}}{2020}]%
        {mo2020darknetz}
\bibfield{author}{\bibinfo{person}{Fan Mo}, \bibinfo{person}{Ali~Shahin
  Shamsabadi}, \bibinfo{person}{Kleomenis Katevas}, \bibinfo{person}{Soteris
  Demetriou}, \bibinfo{person}{Ilias Leontiadis}, \bibinfo{person}{Andrea
  Cavallaro}, {and} \bibinfo{person}{Hamed Haddadi}.}
  \bibinfo{year}{2020}\natexlab{}.
\newblock \showarticletitle{Darknetz: towards model privacy at the edge using
  trusted execution environments}. In \bibinfo{booktitle}{\emph{Proceedings of
  the 18th International Conference on Mobile Systems, Applications, and
  Services}}. \bibinfo{pages}{161--174}.
\newblock


\bibitem[\protect\citeauthoryear{Mohassel and Rindal}{Mohassel and
  Rindal}{2018}]%
        {mohassel2018aby3}
\bibfield{author}{\bibinfo{person}{Payman Mohassel} {and}
  \bibinfo{person}{Peter Rindal}.} \bibinfo{year}{2018}\natexlab{}.
\newblock \showarticletitle{ABY3: A mixed protocol framework for machine
  learning}. In \bibinfo{booktitle}{\emph{Proceedings of the 2018 ACM SIGSAC
  Conference on Computer and Communications Security}}.
  \bibinfo{pages}{35--52}.
\newblock


\bibitem[\protect\citeauthoryear{Mohassel and Zhang}{Mohassel and
  Zhang}{2017}]%
        {mohassel2017secureml}
\bibfield{author}{\bibinfo{person}{Payman Mohassel} {and}
  \bibinfo{person}{Yupeng Zhang}.} \bibinfo{year}{2017}\natexlab{}.
\newblock \showarticletitle{Secureml: A system for scalable privacy-preserving
  machine learning}. In \bibinfo{booktitle}{\emph{2017 IEEE Symposium on
  Security and Privacy (SP)}}. IEEE, \bibinfo{pages}{19--38}.
\newblock


\bibitem[\protect\citeauthoryear{Narayanan, Harlap, Phanishayee, Seshadri,
  Devanur, Ganger, Gibbons, and Zaharia}{Narayanan et~al\mbox{.}}{2019}]%
        {narayanan2019pipedream}
\bibfield{author}{\bibinfo{person}{Deepak Narayanan}, \bibinfo{person}{Aaron
  Harlap}, \bibinfo{person}{Amar Phanishayee}, \bibinfo{person}{Vivek
  Seshadri}, \bibinfo{person}{Nikhil~R Devanur}, \bibinfo{person}{Gregory~R
  Ganger}, \bibinfo{person}{Phillip~B Gibbons}, {and} \bibinfo{person}{Matei
  Zaharia}.} \bibinfo{year}{2019}\natexlab{}.
\newblock \showarticletitle{PipeDream: generalized pipeline parallelism for DNN
  training}. In \bibinfo{booktitle}{\emph{Proceedings of the 27th ACM Symposium
  on Operating Systems Principles}}. \bibinfo{pages}{1--15}.
\newblock


\bibitem[\protect\citeauthoryear{Narra, Lin, Kiamari, Avestimehr, and
  Annavaram}{Narra et~al\mbox{.}}{2019a}]%
        {narra2019slack}
\bibfield{author}{\bibinfo{person}{Krishna~Giri Narra},
  \bibinfo{person}{Zhifeng Lin}, \bibinfo{person}{Mehrdad Kiamari},
  \bibinfo{person}{Salman Avestimehr}, {and} \bibinfo{person}{Murali
  Annavaram}.} \bibinfo{year}{2019}\natexlab{a}.
\newblock \showarticletitle{Slack squeeze coded computing for adaptive
  straggler mitigation}. In \bibinfo{booktitle}{\emph{Proceedings of the
  International Conference for High Performance Computing, Networking, Storage
  and Analysis}}. \bibinfo{pages}{1--16}.
\newblock


\bibitem[\protect\citeauthoryear{Narra, Lin, Wang, Balasubramaniam, and
  Annavaram}{Narra et~al\mbox{.}}{2019b}]%
        {narra2019privacy}
\bibfield{author}{\bibinfo{person}{Krishna~Giri Narra},
  \bibinfo{person}{Zhifeng Lin}, \bibinfo{person}{Yongqin Wang},
  \bibinfo{person}{Keshav Balasubramaniam}, {and} \bibinfo{person}{Murali
  Annavaram}.} \bibinfo{year}{2019}\natexlab{b}.
\newblock \showarticletitle{Privacy-Preserving Inference in Machine Learning
  Services Using Trusted Execution Environments}.
\newblock \bibinfo{journal}{\emph{arXiv preprint arXiv:1912.03485}}
  (\bibinfo{year}{2019}).
\newblock


\bibitem[\protect\citeauthoryear{Ng, Chow, Woo, Wong, and Zhao}{Ng
  et~al\mbox{.}}{2021}]%
        {ng2021goten}
\bibfield{author}{\bibinfo{person}{Lucien~KL Ng}, \bibinfo{person}{Sherman~SM
  Chow}, \bibinfo{person}{Anna~PY Woo}, \bibinfo{person}{Donald~PH Wong}, {and}
  \bibinfo{person}{Yongjun Zhao}.} \bibinfo{year}{2021}\natexlab{}.
\newblock \showarticletitle{Goten: GPU-Outsourcing Trusted Execution of Neural
  Network Training}. In \bibinfo{booktitle}{\emph{Proceedings of the AAAI
  Conference on Artificial Intelligence}}, Vol.~\bibinfo{volume}{35}.
  \bibinfo{pages}{14876--14883}.
\newblock


\bibitem[\protect\citeauthoryear{Oleksenko, Trach, Krahn, Silberstein, and
  Fetzer}{Oleksenko et~al\mbox{.}}{2018}]%
        {oleksenko2018varys}
\bibfield{author}{\bibinfo{person}{Oleksii Oleksenko}, \bibinfo{person}{Bohdan
  Trach}, \bibinfo{person}{Robert Krahn}, \bibinfo{person}{Mark Silberstein},
  {and} \bibinfo{person}{Christof Fetzer}.} \bibinfo{year}{2018}\natexlab{}.
\newblock \showarticletitle{Varys: Protecting $\{$SGX$\}$ enclaves from
  practical side-channel attacks}. In \bibinfo{booktitle}{\emph{2018
  $\{$Usenix$\}$ Annual Technical Conference ($\{$USENIX$\}$$\{$ATC$\}$ 18)}}.
  \bibinfo{pages}{227--240}.
\newblock


\bibitem[\protect\citeauthoryear{Park, Kang, Kim, Kwon, and Huh}{Park
  et~al\mbox{.}}{2020}]%
        {park2020nested}
\bibfield{author}{\bibinfo{person}{Joongun Park}, \bibinfo{person}{Naegyeong
  Kang}, \bibinfo{person}{Taehoon Kim}, \bibinfo{person}{Youngjin Kwon}, {and}
  \bibinfo{person}{Jaehyuk Huh}.} \bibinfo{year}{2020}\natexlab{}.
\newblock \showarticletitle{Nested enclave: supporting fine-grained
  hierarchical isolation with SGX}. In \bibinfo{booktitle}{\emph{2020 ACM/IEEE
  47th Annual International Symposium on Computer Architecture (ISCA)}}. IEEE,
  \bibinfo{pages}{776--789}.
\newblock


\bibitem[\protect\citeauthoryear{Pfister}{Pfister}{2001}]%
        {pfister2001introduction}
\bibfield{author}{\bibinfo{person}{Gregory~F Pfister}.}
  \bibinfo{year}{2001}\natexlab{}.
\newblock \showarticletitle{An introduction to the infiniband architecture}.
\newblock \bibinfo{journal}{\emph{High performance mass storage and parallel
  I/O}} \bibinfo{volume}{42}, \bibinfo{number}{617-632} (\bibinfo{year}{2001}),
  \bibinfo{pages}{102}.
\newblock


\bibitem[\protect\citeauthoryear{Prakash, Hashemi, Wang, Annavaram, and
  Avestimehr}{Prakash et~al\mbox{.}}{2020}]%
        {prakash2020mitigating}
\bibfield{author}{\bibinfo{person}{Saurav Prakash}, \bibinfo{person}{Hanieh
  Hashemi}, \bibinfo{person}{Yongqin Wang}, \bibinfo{person}{Murali Annavaram},
  {and} \bibinfo{person}{Amir~Salman Avestimehr}.}
  \bibinfo{year}{2020}\natexlab{}.
\newblock \showarticletitle{Mitigating byzantine attacks in federated
  learning}.
\newblock \bibinfo{journal}{\emph{arXiv preprint arXiv:2010.07541}}
  (\bibinfo{year}{2020}).
\newblock


\bibitem[\protect\citeauthoryear{Rajat, Wang, and Annavaram}{Rajat
  et~al\mbox{.}}{2021}]%
        {laoram}
\bibfield{author}{\bibinfo{person}{Rachit Rajat}, \bibinfo{person}{Yongqin
  Wang}, {and} \bibinfo{person}{Murali Annavaram}.}
  \bibinfo{year}{2021}\natexlab{}.
\newblock \bibinfo{title}{LAORAM: A Look Ahead ORAM Architecture for Training
  Large Embedding Tables}.
\newblock
\newblock
\urldef\tempurl%
\url{https://doi.org/10.48550/ARXIV.2107.08094}
\showDOI{\tempurl}


\bibitem[\protect\citeauthoryear{Rani, Nayak, and Vyas}{Rani
  et~al\mbox{.}}{2015}]%
        {rani2015ontology}
\bibfield{author}{\bibinfo{person}{Monika Rani}, \bibinfo{person}{Riju Nayak},
  {and} \bibinfo{person}{OP Vyas}.} \bibinfo{year}{2015}\natexlab{}.
\newblock \showarticletitle{An ontology-based adaptive personalized e-learning
  system, assisted by software agents on cloud storage}.
\newblock \bibinfo{journal}{\emph{Knowledge-Based Systems}}
  \bibinfo{volume}{90} (\bibinfo{year}{2015}), \bibinfo{pages}{33--48}.
\newblock


\bibitem[\protect\citeauthoryear{Riazi, Rouani, and Koushanfar}{Riazi
  et~al\mbox{.}}{2019}]%
        {riazi2019deep}
\bibfield{author}{\bibinfo{person}{M~Sadegh Riazi},
  \bibinfo{person}{Bita~Darvish Rouani}, {and} \bibinfo{person}{Farinaz
  Koushanfar}.} \bibinfo{year}{2019}\natexlab{}.
\newblock \showarticletitle{Deep learning on private data}.
\newblock \bibinfo{journal}{\emph{IEEE Security \& Privacy}}
  \bibinfo{volume}{17}, \bibinfo{number}{6} (\bibinfo{year}{2019}),
  \bibinfo{pages}{54--63}.
\newblock


\bibitem[\protect\citeauthoryear{Russakovsky, Deng, Su, Krause, Satheesh, Ma,
  Huang, Karpathy, Khosla, Bernstein, et~al\mbox{.}}{Russakovsky
  et~al\mbox{.}}{2015}]%
        {russakovsky2015imagenet}
\bibfield{author}{\bibinfo{person}{Olga Russakovsky}, \bibinfo{person}{Jia
  Deng}, \bibinfo{person}{Hao Su}, \bibinfo{person}{Jonathan Krause},
  \bibinfo{person}{Sanjeev Satheesh}, \bibinfo{person}{Sean Ma},
  \bibinfo{person}{Zhiheng Huang}, \bibinfo{person}{Andrej Karpathy},
  \bibinfo{person}{Aditya Khosla}, \bibinfo{person}{Michael Bernstein},
  {et~al\mbox{.}}} \bibinfo{year}{2015}\natexlab{}.
\newblock \showarticletitle{Imagenet large scale visual recognition challenge}.
\newblock \bibinfo{journal}{\emph{International journal of computer vision}}
  \bibinfo{volume}{115}, \bibinfo{number}{3} (\bibinfo{year}{2015}),
  \bibinfo{pages}{211--252}.
\newblock


\bibitem[\protect\citeauthoryear{Sandler, Howard, Zhu, Zhmoginov, and
  Chen}{Sandler et~al\mbox{.}}{2018}]%
        {sandler2018mobilenetv2}
\bibfield{author}{\bibinfo{person}{Mark Sandler}, \bibinfo{person}{Andrew
  Howard}, \bibinfo{person}{Menglong Zhu}, \bibinfo{person}{Andrey Zhmoginov},
  {and} \bibinfo{person}{Liang-Chieh Chen}.} \bibinfo{year}{2018}\natexlab{}.
\newblock \showarticletitle{Mobilenetv2: Inverted residuals and linear
  bottlenecks}. In \bibinfo{booktitle}{\emph{Proceedings of the IEEE conference
  on computer vision and pattern recognition}}. \bibinfo{pages}{4510--4520}.
\newblock


\bibitem[\protect\citeauthoryear{Schwartz}{Schwartz}{2021}]%
        {SGXMemory}
\bibfield{author}{\bibinfo{person}{Jeffrey Schwartz}.}
  \bibinfo{year}{2021}\natexlab{}.
\newblock \bibinfo{booktitle}{\emph{Intel Makes 3rd Gen Xeon Scalable Processor
  Rollout Official}}.
\newblock
\urldef\tempurl%
\url{https://www.channelfutures.com/data-centers/intel-makes-3rd-gen-xeon-scalable-processor-rollout-official}
\showURL{%
\tempurl}


\bibitem[\protect\citeauthoryear{Shanley}{Shanley}{2003}]%
        {shanley2003infiniband}
\bibfield{author}{\bibinfo{person}{Tom Shanley}.}
  \bibinfo{year}{2003}\natexlab{}.
\newblock \bibinfo{booktitle}{\emph{InfiniBand network architecture}}.
\newblock \bibinfo{publisher}{Addison-Wesley Professional}.
\newblock


\bibitem[\protect\citeauthoryear{Shokri and Shmatikov}{Shokri and
  Shmatikov}{2015}]%
        {shokri2015privacy}
\bibfield{author}{\bibinfo{person}{Reza Shokri} {and} \bibinfo{person}{Vitaly
  Shmatikov}.} \bibinfo{year}{2015}\natexlab{}.
\newblock \showarticletitle{Privacy-preserving deep learning}. In
  \bibinfo{booktitle}{\emph{Proceedings of the 22nd ACM SIGSAC conference on
  computer and communications security}}. \bibinfo{pages}{1310--1321}.
\newblock


\bibitem[\protect\citeauthoryear{Simonyan and Zisserman}{Simonyan and
  Zisserman}{2014}]%
        {simonyan2014very}
\bibfield{author}{\bibinfo{person}{Karen Simonyan} {and}
  \bibinfo{person}{Andrew Zisserman}.} \bibinfo{year}{2014}\natexlab{}.
\newblock \showarticletitle{Very deep convolutional networks for large-scale
  image recognition}.
\newblock \bibinfo{journal}{\emph{arXiv preprint arXiv:1409.1556}}
  (\bibinfo{year}{2014}).
\newblock


\bibitem[\protect\citeauthoryear{So, Guler, and Avestimehr}{So
  et~al\mbox{.}}{2020}]%
        {so2020byzantine}
\bibfield{author}{\bibinfo{person}{Jinhyun So}, \bibinfo{person}{Basak Guler},
  {and} \bibinfo{person}{A~Salman Avestimehr}.}
  \bibinfo{year}{2020}\natexlab{}.
\newblock \showarticletitle{Byzantine-Resilient Secure Federated Learning}.
\newblock \bibinfo{journal}{\emph{arXiv preprint arXiv:2007.11115}}
  (\bibinfo{year}{2020}).
\newblock


\bibitem[\protect\citeauthoryear{Spruill}{Spruill}{1983}]%
        {spruill1983confidentiality}
\bibfield{author}{\bibinfo{person}{Nancy Spruill}.}
  \bibinfo{year}{1983}\natexlab{}.
\newblock \showarticletitle{The confidentiality and analytic usefulness of
  masked business microdata}.
\newblock \bibinfo{journal}{\emph{Proceedings of the Section on Survey Research
  Methods, 1983}} (\bibinfo{year}{1983}), \bibinfo{pages}{602--607}.
\newblock


\bibitem[\protect\citeauthoryear{Tang, Ali, Hashemi, Gangwani, Avestimehr, and
  Annavaram}{Tang et~al\mbox{.}}{2021}]%
        {tang2021verifiable}
\bibfield{author}{\bibinfo{person}{Tingting Tang}, \bibinfo{person}{Ramy~E
  Ali}, \bibinfo{person}{Hanieh Hashemi}, \bibinfo{person}{Tynan Gangwani},
  \bibinfo{person}{Salman Avestimehr}, {and} \bibinfo{person}{Murali
  Annavaram}.} \bibinfo{year}{2021}\natexlab{}.
\newblock \showarticletitle{Verifiable coded computing: Towards fast, secure
  and private distributed machine learning}.
\newblock \bibinfo{journal}{\emph{arXiv preprint arXiv:2107.12958}}
  (\bibinfo{year}{2021}).
\newblock


\bibitem[\protect\citeauthoryear{Team}{Team}{2017}]%
        {team2017learning}
\bibfield{author}{\bibinfo{person}{ADP Team}.} \bibinfo{year}{2017}\natexlab{}.
\newblock \showarticletitle{Learning with privacy at scale}.
\newblock \bibinfo{journal}{\emph{Apple Mach. Learn. J}} \bibinfo{volume}{1},
  \bibinfo{number}{9} (\bibinfo{year}{2017}).
\newblock


\bibitem[\protect\citeauthoryear{team}{team}{2021}]%
        {intel}
\bibfield{author}{\bibinfo{person}{SGX team}.} \bibinfo{year}{2021}\natexlab{}.
\newblock \bibinfo{booktitle}{\emph{Intel SGX in clouds}}.
\newblock
\urldef\tempurl%
\url{https://software.intel.com/content/www/us/en/develop/topics/software-guard-extensions.html}
\showURL{%
\tempurl}


\bibitem[\protect\citeauthoryear{Tramer and Boneh}{Tramer and Boneh}{2018}]%
        {tramer2018slalom}
\bibfield{author}{\bibinfo{person}{Florian Tramer} {and} \bibinfo{person}{Dan
  Boneh}.} \bibinfo{year}{2018}\natexlab{}.
\newblock \showarticletitle{Slalom: Fast, Verifiable and Private Execution of
  Neural Networks in Trusted Hardware}. In
  \bibinfo{booktitle}{\emph{International Conference on Learning
  Representations}}.
\newblock


\bibitem[\protect\citeauthoryear{Volos, Vaswani, and Bruno}{Volos
  et~al\mbox{.}}{2018}]%
        {volos2018graviton}
\bibfield{author}{\bibinfo{person}{Stavros Volos}, \bibinfo{person}{Kapil
  Vaswani}, {and} \bibinfo{person}{Rodrigo Bruno}.}
  \bibinfo{year}{2018}\natexlab{}.
\newblock \showarticletitle{Graviton: Trusted execution environments on gpus}.
  In \bibinfo{booktitle}{\emph{13th $\{$USENIX$\}$ Symposium on Operating
  Systems Design and Implementation ($\{$OSDI$\}$ 18)}}.
  \bibinfo{pages}{681--696}.
\newblock


\bibitem[\protect\citeauthoryear{Wagh, Gupta, and Chandran}{Wagh
  et~al\mbox{.}}{2019}]%
        {wagh2019securenn}
\bibfield{author}{\bibinfo{person}{Sameer Wagh}, \bibinfo{person}{Divya Gupta},
  {and} \bibinfo{person}{Nishanth Chandran}.} \bibinfo{year}{2019}\natexlab{}.
\newblock \showarticletitle{Securenn: 3-party secure computation for neural
  network training}.
\newblock \bibinfo{journal}{\emph{Proceedings on Privacy Enhancing
  Technologies}} \bibinfo{volume}{2019}, \bibinfo{number}{3}
  (\bibinfo{year}{2019}), \bibinfo{pages}{26--49}.
\newblock


\bibitem[\protect\citeauthoryear{Wang, Chen, Pan, Zhang, Wang, Bindschaedler,
  Tang, and Gunter}{Wang et~al\mbox{.}}{2017}]%
        {wang2017leaky}
\bibfield{author}{\bibinfo{person}{Wenhao Wang}, \bibinfo{person}{Guoxing
  Chen}, \bibinfo{person}{Xiaorui Pan}, \bibinfo{person}{Yinqian Zhang},
  \bibinfo{person}{XiaoFeng Wang}, \bibinfo{person}{Vincent Bindschaedler},
  \bibinfo{person}{Haixu Tang}, {and} \bibinfo{person}{Carl~A Gunter}.}
  \bibinfo{year}{2017}\natexlab{}.
\newblock \showarticletitle{Leaky cauldron on the dark land: Understanding
  memory side-channel hazards in SGX}. In \bibinfo{booktitle}{\emph{Proceedings
  of the 2017 ACM SIGSAC Conference on Computer and Communications Security}}.
  \bibinfo{pages}{2421--2434}.
\newblock


\bibitem[\protect\citeauthoryear{Wang, Suh, Xiong, Lefaudeux, Knott, Annavaram,
  and Lee}{Wang et~al\mbox{.}}{2022}]%
        {ispass2022wang}
\bibfield{author}{\bibinfo{person}{Yongqin Wang}, \bibinfo{person}{G.~Edward
  Suh}, \bibinfo{person}{Wenjie Xiong}, \bibinfo{person}{Benjamin Lefaudeux},
  \bibinfo{person}{Brian Knott}, \bibinfo{person}{Murali Annavaram}, {and}
  \bibinfo{person}{Hsien-Hsin~S. Lee}.} \bibinfo{year}{2022}\natexlab{}.
\newblock \showarticletitle{Characterization of MPC-based Private Inference for
  Transformer-based Models}. In \bibinfo{booktitle}{\emph{2022 IEEE
  International Symposium on Performance Analysis of Systems and Software
  (ISPASS)}}. \bibinfo{publisher}{IEEE Computer Society}, \bibinfo{address}{Los
  Alamitos, CA, USA}, \bibinfo{pages}{187--197}.
\newblock
\urldef\tempurl%
\url{https://doi.org/10.1109/ISPASS55109.2022.00025}
\showDOI{\tempurl}


\bibitem[\protect\citeauthoryear{Xu, Cui, and Peinado}{Xu
  et~al\mbox{.}}{2015}]%
        {xu2015controlled}
\bibfield{author}{\bibinfo{person}{Yuanzhong Xu}, \bibinfo{person}{Weidong
  Cui}, {and} \bibinfo{person}{Marcus Peinado}.}
  \bibinfo{year}{2015}\natexlab{}.
\newblock \showarticletitle{Controlled-channel attacks: Deterministic side
  channels for untrusted operating systems}. In \bibinfo{booktitle}{\emph{2015
  IEEE Symposium on Security and Privacy}}. IEEE, \bibinfo{pages}{640--656}.
\newblock


\bibitem[\protect\citeauthoryear{Xu, Mauldin, Yao, Pei, Wei, and Yang}{Xu
  et~al\mbox{.}}{2020}]%
        {xu2020bus}
\bibfield{author}{\bibinfo{person}{Zhenyu Xu}, \bibinfo{person}{Thomas
  Mauldin}, \bibinfo{person}{Zheyi Yao}, \bibinfo{person}{Shuyi Pei},
  \bibinfo{person}{Tao Wei}, {and} \bibinfo{person}{Qing Yang}.}
  \bibinfo{year}{2020}\natexlab{}.
\newblock \showarticletitle{A bus authentication and anti-probing architecture
  extending hardware trusted computing base off CPU chips and beyond}. In
  \bibinfo{booktitle}{\emph{2020 ACM/IEEE 47th Annual International Symposium
  on Computer Architecture (ISCA)}}. IEEE, \bibinfo{pages}{749--761}.
\newblock


\bibitem[\protect\citeauthoryear{Yang, Zhang, Kirichenko, Bai, Wilson, and
  De~Sa}{Yang et~al\mbox{.}}{2019}]%
        {yang2019swalp}
\bibfield{author}{\bibinfo{person}{Guandao Yang}, \bibinfo{person}{Tianyi
  Zhang}, \bibinfo{person}{Polina Kirichenko}, \bibinfo{person}{Junwen Bai},
  \bibinfo{person}{Andrew~Gordon Wilson}, {and} \bibinfo{person}{Chris De~Sa}.}
  \bibinfo{year}{2019}\natexlab{}.
\newblock \showarticletitle{SWALP: Stochastic weight averaging in low precision
  training}. In \bibinfo{booktitle}{\emph{International Conference on Machine
  Learning}}. PMLR, \bibinfo{pages}{7015--7024}.
\newblock


\bibitem[\protect\citeauthoryear{Yu, Li, Raviv, Kalan, Soltanolkotabi, and
  Avestimehr}{Yu et~al\mbox{.}}{2019}]%
        {yu2019lagrange}
\bibfield{author}{\bibinfo{person}{Qian Yu}, \bibinfo{person}{Songze Li},
  \bibinfo{person}{Netanel Raviv}, \bibinfo{person}{Seyed Mohammadreza~Mousavi
  Kalan}, \bibinfo{person}{Mahdi Soltanolkotabi}, {and}
  \bibinfo{person}{Salman~A Avestimehr}.} \bibinfo{year}{2019}\natexlab{}.
\newblock \showarticletitle{Lagrange coded computing: Optimal design for
  resiliency, security, and privacy}. In \bibinfo{booktitle}{\emph{The 22nd
  International Conference on Artificial Intelligence and Statistics}}. PMLR,
  \bibinfo{pages}{1215--1225}.
\newblock


\bibitem[\protect\citeauthoryear{Zhang, Gupta, Lian, and Liu}{Zhang
  et~al\mbox{.}}{2015}]%
        {zhang2015staleness}
\bibfield{author}{\bibinfo{person}{Wei Zhang}, \bibinfo{person}{Suyog Gupta},
  \bibinfo{person}{Xiangru Lian}, {and} \bibinfo{person}{Ji Liu}.}
  \bibinfo{year}{2015}\natexlab{}.
\newblock \showarticletitle{Staleness-aware async-sgd for distributed deep
  learning}.
\newblock \bibinfo{journal}{\emph{arXiv preprint arXiv:1511.05950}}
  (\bibinfo{year}{2015}).
\newblock


\bibitem[\protect\citeauthoryear{Zhu, Ferguson, and Dolgov}{Zhu
  et~al\mbox{.}}{2014}]%
        {zhu2014system}
\bibfield{author}{\bibinfo{person}{Jiajun Zhu}, \bibinfo{person}{David~I
  Ferguson}, {and} \bibinfo{person}{Dmitri~A Dolgov}.}
  \bibinfo{year}{2014}\natexlab{}.
\newblock \bibinfo{title}{System and method for predicting behaviors of
  detected objects}.
\newblock
\newblock
\newblock
\shownote{US Patent 8,660,734}.


\bibitem[\protect\citeauthoryear{Zhu, Liu, and Han}{Zhu et~al\mbox{.}}{2019}]%
        {zhu2019deep}
\bibfield{author}{\bibinfo{person}{Ligeng Zhu}, \bibinfo{person}{Zhijian Liu},
  {and} \bibinfo{person}{Song Han}.} \bibinfo{year}{2019}\natexlab{}.
\newblock \showarticletitle{Deep leakage from gradients}. In
  \bibinfo{booktitle}{\emph{Advances in Neural Information Processing
  Systems}}. \bibinfo{pages}{14747--14756}.
\newblock


\end{thebibliography}

\appendix

\end{document}